\def\preprint{0}
\ifnum\preprint=9
    \documentclass[11pt,draft,a4paper]{article} 
\fi 
\def\masterthesis{0}
\def\bigfont{0} 
\def\cryptology{0} 
\def\sigconf{0} 
\ifnum\sigconf=0
   \RequirePackage{etex}
\fi
\def\llncs{1} 
\def\lipics{0} 
\def\quantumjournal{0}

\def\draft{0}
\def\anonymous{0} 
\def\shownomenclature{1} 
\def\smallbib{0} 
\def\toc{0} 
\ifnum
    \ifnum\preprint=1 1\else\ifnum\cryptology=1 1\else 0\fi\fi
=1 
    \ifnum\draft=1
       \documentclass[11pt,draft,a4paper]{article}
       \pdfoutput=1
    \else
        \documentclass[11pt,final,a4paper]{article}
        \pdfoutput=1
    \fi
    \usepackage{lmodern}
    \ifnum\preprint=1
        \usepackage{fullpage}
    \fi
\fi
\ifnum\masterthesis=1
    \documentclass[11pt,final,a4paper,titlepage]{article}
    \usepackage{lmodern}
\fi
\ifnum\sigconf=1
        \documentclass[sigconf,final]{acmart}
\fi
\ifnum\bigfont=1
    \documentclass[final,17pt]{extarticle} \usepackage{fullpage}
\fi
\ifnum\llncs=1
    \ifnum\draft=0
        \documentclass[runningheads]{llncs}
    \else
        \documentclass[runningheads,draft]{llncs}
    \fi
\fi
\ifnum\lipics=1
    \input{lipics_header}
\fi
\ifnum\quantumjournal=1
  \documentclass[a4paper,onecolumn,11pt 
  ]
  {quantumarticle}
  \pdfoutput=1

\fi

\ifnum\toc=1\ifnum\lipics=0
    \usepackage[nottoc]{tocbibind} 
\fi\fi
\ifnum\lipics=0\ifnum\sigconf=0\usepackage{authblk}\fi\fi 
\ifnum\sigconf=0
    \usepackage{silence} 
    \usepackage{amssymb} 
\fi
\ifnum\lipics=0 
    \usepackage[modulo]{lineno}
\fi

\ifnum\llncs=1

\fi \usepackage{amsthm}

\usepackage{amsfonts,latexsym,color,paralist,url,ifdraft,mathrsfs,thm-restate,comment,bbold,booktabs,ifthen,wrapfig}

\usepackage[utf8]{inputenc} \usepackage[autostyle=false, style=english]{csquotes} \MakeOuterQuote{"}

\usepackage[T1]{fontenc}
\ifnum\sigconf=0
    \ifnum \shownomenclature=1
        \usepackage[refpage,prefix]{nomencl}
         
        \makenomenclature
        
  \makeatletter
        \def\thenomenclature{%
          \providecommand*{\listofnlsname}{\nomname}%

              \section{\nomname}
              \label{sec:nomenclature}    
              \if@intoc\addcontentsline{toc}{section}{\nomname}\fi%

          \nompreamble
          \if@nomentbl
            \let\itemOrig=\item
            \def\item{\gdef\item{\\}}%
            \expandafter\longtable\expandafter{\@nomtableformat}
          \else
            \list{}{%
              \labelwidth\nom@tempdim
              \leftmargin\labelwidth
              \advance\leftmargin\labelsep
              \itemsep\nomitemsep
              \let\makelabel\nomlabel}%
          \fi
        }
        \makeatother 
    \fi
    \ifnum\lipics=0    
        \usepackage[draft=false,pageanchor]{hyperref}
        \usepackage[final]{graphicx}
    \fi
    \usepackage{doi}
\fi

\ifdraft{\usepackage{showlabels}}{} 

\usepackage [ lambda,advantage , operators , sets , adversary , landau , probability , notions , logic ,ff, mm,primitives , events , complexity , asymptotics , keys]{cryptocode}

\ifdraft{\newcommand{\authnote}[3]{{\color{#3} {\bf  #1:} #2}}}{\newcommand{\authnote}[3]{}}

\newcommand{\tensor}{\otimes}

\ifdraft{\linenumbers}{}

\usepackage{algorithm,algorithmicx,algpseudocode}
\usepackage[normalem]{ulem}

\usepackage{tikz}
\usetikzlibrary{arrows,chains,matrix,positioning,scopes}
\makeatletter
\tikzset{join/.code=\tikzset{after node path={%
\ifx\tikzchainprevious\pgfutil@empty\else(\tikzchainprevious)%
edge[every join]#1(\tikzchaincurrent)\fi}}}
\makeatother
\tikzset{>=stealth',every on chain/.append style={join},
         every join/.style={->}}
\tikzstyle{labeled}=[execute at begin node=$\scriptstyle,
   execute at end node=$]
\usepackage[capitalise]{cleveref} 
\crefname{Game}{Game}{Games}

\ifthenelse{
\(\equal{\cryptology}{1}\OR
\(\equal{\preprint}{1} \OR \equal{\masterthesis}{1}\) \OR \(\equal{\bigfont}{1} \OR \equal{\quantumjournal}{1} \)
\) 
}
{
    \newtheorem{theorem}{Theorem} 
    \newtheorem{lemma}{Lemma}
    \newtheorem{corollary}{Corollary}
    \newtheorem{proposition}{Proposition}
    \newtheorem{definition}{Definition}

    \newtheorem{claim}{Claim}
    \newtheorem{fact}{Fact}
    \newtheorem*{theorem*}{Theorem}
    \newtheorem*{lemma*}{Lemma}
    \newtheorem*{corollary*}{Corollary}
    \newtheorem*{proposition*}{Proposition}
    \newtheorem*{claim*}{Claim}
    \theoremstyle{definition}
    \newtheorem{openproblem}{Open Problem}
    \theoremstyle{remark}
    \newtheorem{remark}[theorem]{Remark}
    \newtheorem{construction}{Construction}
    
    \theoremstyle{plain}
    
}{}  
\ifnum\sigconf=1

  \newtheorem*{theorem*}{Theorem}
  \newtheorem*{lemma*}{Lemma}
  \newtheorem*{corollary*}{Corollary}
  \newtheorem*{proposition*}{Proposition}
  \newtheorem*{claim*}{Claim}
  \theoremstyle{definition}

\fi
\ifnum\llncs=1
    \spnewtheorem{construction}{Construction}{\bfseries}{\itshape}
    \crefname{construction}{Construction}{Constructions}
    \Crefname{construction}{Construction}{Constructions}

\fi
  
\theoremstyle{plain}
\newcounter{Game} 
\expandafter\let\expandafter\savedflalignstar\csname flalign*\endcsname
\expandafter\let\expandafter\savedendflalignstar\csname endflalign*\endcsname
\AtBeginDocument{%
  \expandafter\let\csname flalign*\endcsname\savedflalignstar
  \expandafter\let\csname endflalign*\endcsname\savedendflalignstar
}

\newcommand{\ab}[1]{}  

\ifnum\masterthesis=1
    \usepackage{setspace}
    \usepackage{fancyhdr}
\fi

\usepackage{xspace}
\DeclareMathAlphabet{\mathpzc}{OT1}{pzc}{m}{it}

\newcommand{\oss}{\ensuremath{\mathsf{OSS}}\xspace}

\newcommand{\qsk}{\ensuremath{\mathpzc{sk}}\xspace} 
\newcommand{\quant}[1]{\ensuremath\mathpzc{{#1}}\xspace}

\newcommand{\spark}{\mathpzc{Spark}\xspace}
\newcommand{\clone}{\mathpzc{Clone}\xspace}
\newcommand{\recon}{\mathpzc{Reconstruct}\xspace}
\newcommand{\prepare}{\mathpzc{Prepare}\xspace}
\newcommand{\decon}{\mathpzc{Deconstruct}\xspace}
\newcommand{\ver}{\mathpzc{Verify}\xspace}
\newcommand{\conv}{\mathpzc{Converse}\xspace}

\newcommand{\tildever}{\widetilde{\mathpzc{Verify}}\xspace}

\newcommand{\oneshotgen}{\mathpzc{KeyGen}\xspace}
\newcommand{\oneshotsign}{\mathpzc{Sign\xspace}}
\newcommand{\oneshotversign}{\textsf{VerSign}\xspace}
\newcommand{\oneshotvertoken}{\mathpzc{VerToken}\xspace}

\newcommand{\oneshotcgen}{\textsf{CKeygen}\xspace}

\newcommand{\certgenent}{\mathpzc{GenEntropy}\xspace}
\newcommand{\certver}{\mathpzc{CertEntropy}\xspace}
\newcommand{\certclonewitness}{\mathpzc{CloneWitness}\xspace}
\newcommand{\certprover}{\mathpzc{Prover}\xspace}
\newcommand{\certcertifer}{\mathpzc{Certifier}\xspace}
\newcommand{\certcver}{\textsf{CCertEntropy}\xspace}
\newcommand{\ifprepare}{\textsf{CPrepare}\xspace}
\newcommand{\ifver}{\textsf{CVerify}\xspace}
\newcommand{\ifqver}{\mathpzc{QVerify}\xspace}
\newcommand{\ifrecon}{\textsf{CRecon}\xspace}

\newcommand{\bigstack}[1]{\vcenter{\hbox{$\begin{gathered}#1\end{gathered}$}}}

\usepackage{tikz}
\usetikzlibrary{quantikz2,arrows.meta,backgrounds,calc,fadings}
\usepackage[operators,sets]{cryptocode}
\usepackage{centernot}
\usepackage{multirow}
\usepackage{multicol}
\usepackage{graphicx} 
\usepackage{array}    
\usepackage{enumitem}
\usepackage{pifont}
\usepackage{colortbl}
\usepackage{float}
\newcommand{\cmark}{\textcolor{green!55!black}{\ding{51}}}
\newcommand{\xmark}{\textcolor{red!75!black}{\ding{55}}}
\DeclareMathOperator*{\Hinf}{H_{\infty}}
\DeclareMathOperator*{\Hbot}{H_{\centernot\bot,\infty}}
\begin{document}
\ifnum\masterthesis=0 
    \title{Conversable Quantum Fire in the Standard Model  
    }
    \date{}
\fi

\ifnum\anonymous=0
    \ifnum\preprint=1
        \author[1]{XXX}
        \author[1]{Or Sattath}
        \affil[1]{The Stein Faculty of Computer and Information Science, Ben-Gurion University of the Negev}
    \fi
    \ifnum\cryptology=1
        \author{XXX}
        \affil{The Stein Faculty of Computer and Information Science, Ben-Gurion University of the Negev, Beersheba, Israel\\
        XXX@post.bgu.ac.il}
        \author{Or Sattath}
        \affil{The Stein Faculty of Computer and Information Science, Ben-Gurion University of the Negev, Beersheba, Israel\\
        sattath@bgu.ac.il}
    \fi
    \ifnum\sigconf=1
        \author{Or Sattath}
            \affiliation{%
              \institution{The Stein Faculty of Computer and Information Science, Ben-Gurion University of the Negev}
              \city{Beersheba}
              \country{Israel}}
        \email{sattath@bgu.ac.il}

        \author{XXX}
            \affiliation{%
              \institution{The Stein Faculty of Computer and Information Science, Ben-Gurion University of the Negev}
              \city{Beersheba}
              \country{Israel}}
        \email{XXX@bgu.ac.il}
    \fi
    \ifnum\llncs=1
        %
        \author{Ofer Casper\inst{1}
        \and Barak Nehoran\inst{2}
        \and Or Sattath\inst{1}
        }
        \authorrunning{O. Casper et al.}
        %
        \institute{The Stein Faculty of Computer and Information Science, Ben-Gurion University of the Negev, Beersheba, Israel\\
        \email{caspero@post.bgu.ac.il}, \email{sattath@bgu.ac.il}
        \and
        Department of Computer Science, Columbia University, New York, NY, USA\\
        \email{b.nehoran@columbia.edu}
        }
    \fi
    \ifnum\quantumjournal=1
        \author[1]{Andrea Coladangelo}
        \affil[1]{Computing and Mathematical Sciences,	
        Caltech}
        \author[2]{Or Sattath}
        \email{sattath@bgu.ac.il}

        \affil[2]{The Stein Faculty of Computer and Information Science, Ben-Gurion University of the Negev}
        
    \fi
\else
    \ifnum\lipics=0
        \author{}
    \fi
\fi

\ifthenelse{\equal{\masterthesis}{0} \AND \equal{\sigconf}{0}}{
    \maketitle
} 
{}        
\ifnum\masterthesis=1
    \include{masterthesis_header}
\else 
    \ifnum\sigconf=0
        \begin{abstract}
    \fi
\fi
\ifnum\sigconf=0
    \ifnum\masterthesis=0

    Quantum fire consists of quantum states that can be efficiently functionally cloned but cannot be transmitted using one-way classical communication. Whereas all previous quantum-fire constructions with proven untelegraphability are relative to an oracle, ours, based on one-shot signatures, is in the standard model.

    Our construction achieves two further notions that we introduce: (a) keyless untelegraphability, a strengthening of untelegraphability; and (b) conversability, meaning that a flame, though not telegraphable via one-way classical communication, can be transmitted via classical interaction. Hence one-way and two-way classical communication differ qualitatively in their power to transmit this quantum fire.

    We use these two novel properties of quantum fire, along with its clonability, to give one of the first cryptographic applications of quantum fire. Keyless untelegraphability forces the distribution of serial numbers of valid flames to have high min-entropy, while conversability and clonability allow the corresponding flame to be cloned and transferred using only classical communication.

    We present two variants of our conversable quantum-fire construction. The first, based on one-shot signatures (instantiable from subexponential iO, subexponentially secure one-way functions, and LWE), supports polynomially many flames, and super-logarithmic min-entropy of the serial-number distribution. The second, under the stronger assumption of exponentially unforgeable one-shot signatures (instantiable relative to a classical oracle), supports  exponentially many flames, and linear min-entropy.
    


    \fi
\fi 
\ifnum\masterthesis=0
    \ifnum\sigconf=0
        \end{abstract}
    \fi
\fi

\ifnum\sigconf=1
    \begin{abstract}
        Abstract of SigConf goes here.
    \end{abstract}
    \keywords{}
    \maketitle
\fi

\ifnum\toc=1
    \ifnum\llncs=1
        \setcounter{secnumdepth}{3}
        \setcounter{tocdepth}{3}
    \fi
    \tableofcontents 
\fi

\section{Introduction}

Quantum fire is a quantum primitive consisting of states that are clonable (in a computational sense) yet untelegraphable. We can think of these as states that can spread (be cloned), analogous to "real" fire, provided they are kept "alive" quantumly and do not decohere: the fire cannot be "reconstructed" from the decohered state (the classical ashes), since that would allow it to be telegraphed. 

Quantum fire was introduced by Nehoran and Zhandry~\cite{NZ24} to demonstrate the distinction between two central no-go theorems in quantum information theory: the no-cloning theorem \cite{Park70,WZ82,Dieks82} and the no-telegraphing theorem \cite{Wer98} (originally known as the no-classical teleportation theorem, or no-teleportation without pre-shared entanglement). The no-cloning theorem states that one cannot generate two copies given a
single unknown quantum state, and has proved to be extremely useful in quantum cryptography. 
The no-telegraphing theorem, which is far less commonly used in quantum cryptography, informally shows that an arbitrary quantum state cannot be converted into a classical description from which the state can later be reconstructed. All previous quantum-fire constructions with proven untelegraphability are relative to a given oracle~\cite{NZ24,CGS25}. Bostanci, Nehoran, and Zhandry~\cite{BNZ24} additionally proposed a candidate construction for which untelegraphability remains unproved (for a full discussion, see \cref{table:constructionComparison} and \cref{sec:related_works}). In this work, we give the first provably untelegraphable quantum-fire construction in the standard model, assuming one-shot signatures~\cite{AGKZ20,SZ25}.

Early humans first encountered fire through lightning strikes and forest fires; exploiting it came long before domesticating it, that is, producing and controlling it at will. In this limited respect, existing quantum fire constructions remain undomesticated: their provable untelegraphability depends on an external oracle. The stakes are admittedly lower, but we domesticate quantum fire.

Our construction comes in two variants: a logarithmic construction (see \cref{const:fireFromOneShot}(a)) and a linear construction (see \cref{const:fireFromOneShot}(b)), which differ only in the underlying assumption and in the maximum depth of clone operations allowed.
We view executions as forming a tree of "flames", where the initial flame generated by $\spark$ is at depth $0$, and each application of the $\clone$ operation maps a flame at depth $d$ to two flames at depth $d+1$. 
The logarithmic construction is based on the standard assumption of the existence of an unforgeable one-shot signature scheme, and guarantees security as long as the depth of any flame in the resulting tree is at most $c \log(\secpar)$, for a constant $c$ fixed in advance.
The linear construction, on the other hand, relies on exponential unforgeability of the underlying one-shot signature scheme, and remains secure up to depth $\frac{1}{8}\secpar$. 
We note that in both constructions, the logarithmic and linear names refer to the number of clone operations and hence the depth of the resulting tree; that is, for the logarithmic construction, we can have at most a polynomial number of flames, while for the linear construction, we can have an exponential number of them. The tree structure of our construction is inspired by the construction of ``budget signatures'' in~\cite{AGKZ20}.%
\footnote{
    We thank the authors of~\cite{AGKZ20} for suggesting this idea.
}

We also emphasize that the two quantum flames that result from a $\clone$ operation might be far apart from an information-theoretic or computational perspective (e.g., in trace distance or computational distinguishability; indeed, they can be efficiently distinguished in our construction), but are clones in the sense that both pass verification with the original flame's serial number (see \cref{def:CloningCorrectnessAlt}).

Moreover, we show that our construction satisfies not only the security notion of untelegraphability but a stronger variant, which we term \emph{keyless untelegraphability}. The keyless untelegraphability property, as defined in \cref{def:StrongTelegraphingSecurity}, ensures that no flame can be transmitted via one-way classical communication, even when both the flame and its serial number are generated by an adversary.  
We note that the keyless notion has appeared in other quantum cryptographic schemes: quantum lightning is the keyless analog of public-key quantum money, and one-shot signatures are the keyless analog of tokenized digital signatures. See Ref.~\cite{Sat23} for further discussion of this distinction.
Put together, these yield our first main result:
\begin{theorem}[Main Result 1, formally stated in \cref{coro:OssToConversableFire}]\label{thm:fire_standard_model} 
One-shot signatures (for which standard-model constructions exist) imply keyless untelegraphable quantum fire.
\end{theorem}

Note that keyless untelegraphability is stronger than standard untelegraphability; moreover, as we show in \cref{clm:Min-entropyFire}, keyless untelegraphability implies that the distribution of serial numbers generated by $\spark$ has high min-entropy.

Previous works have considered whether interactive communication enables adversaries to transmit a quantum state over a classical channel.
Specifically, Çakan, Goyal, and Shmueli~\cite{CGS25} prove that their construction satisfies what they call interactive security: even interactive classical communication does not enable the transmission of a flame. In our terminology, their construction is therefore \emph{unconversable}.
We take the opposite perspective and ask whether interaction can be useful: can a quantum fire state be impossible to transmit over one-way classical communication, as required by untelegraphability, yet possible to transmit using interactive classical communication?
To capture this exact phenomenon, we introduce \emph{conversability}, the ability to transmit a quantum state via an interactive classical ``conversation''.
Informally, a quantum fire scheme is conversable if a sender holding a valid flame can enable a receiver to reconstruct a valid flame with the same serial number through an interactive protocol using only classical communication. Thus, conversable quantum fire consists of states that can be cloned and spread via interactive classical communication, but cannot be telegraphed via one-way classical communication.

Our quantum fire construction realizes this new notion:

\begin{theorem} [Main Result 2: Conversable Quantum Fire, informal]
    \label{thm:informalConversability}
    The quantum fire scheme of \cref{const:fireFromOneShot} is conversable: a valid
    flame can be sent using only interactive classical communication, despite
    the untelegraphability of the scheme.
\end{theorem}

Conversability is not merely an additional feature of our construction. We use it to obtain one of the first constructive applications of quantum fire as a building block of another cryptographic primitive: publicly certifiable min-entropy. Previous works have mentioned the prevention of key exfiltration as an application. At first sight, the connection to certifiable min-entropy is unclear: untelegraphability is a property concerning the impossibility of communicating quantum states using classical information, whereas min-entropy is a property of a classical probability distribution. Why, then, should the untelegraphability of a flame imply anything about the entropy of its classical serial number? And even if it does, how can this entropy be certified to—and subsequently by—other parties using only classical communication?

Keyless untelegraphability and conversability, the two main properties of our construction, answer these questions in complementary ways. Keyless untelegraphability forces the serial numbers of the flames to have high min-entropy, while conversability allows the quantum witnesses needed to certify the min-entropy to be transmitted. Now, considering the clonability property of quantum fire, together with conversability, yields another interesting property: transferability.

\paragraph{Transferability.} An important question about public certification is whether it is \emph{transferable}. We say that a scheme is transferable if an accepting verifier obtains a witness that allows it to subsequently act as a prover and certify the same statement to another verifier. 

There are several different variants of transferability, depending on the resources available to the parties. First, one may ask whether transferring the
certification consumes the original prover's witness, or whether the prover
retains the ability to certify the statement to additional verifiers.
Second, transferability may require a quantum communication channel, or may
be achievable using classical communication alone. Finally, the verifier may
be required to be quantum, or certification may be possible with a classical
verifier.
Our construction (see \cref{const:TPCME}) achieves transferability, allowing the prover to retain a valid quantum witness for the min-entropy. Each transfer consumes one level of the underlying fire's cloning-depth budget, so the number of sequential transfers is bounded accordingly. Moreover, the channel needs only to be classical. The verifier, however, is quantum.

Overall, we prove in this work the following theorem:

\begin{theorem}[Main Result 3, formally stated in \cref{coro:OssToTransferablePublicCertifiable}] There is a publicly certifiable, classically transferable min-entropy
scheme, where the number of sequential transfers is bounded by the cloning
depth of the underlying conversable quantum fire.
\label{thm:informalResultPubliclyCertifiableMinEntropy}
\end{theorem}

\begin{remark}
    The construction above requires the verifier to have a quantum computer.
    We also give a modified construction with a classical verifier, under the
    additional assumption that the underlying one-shot signature scheme has
    classically generable verification keys. This modification comes at the cost of transferability: the classical verifier does not obtain the quantum witness needed to subsequently act as a prover.
\end{remark}

As with quantum fire, we can achieve either publicly certifiable super-logarithmic min-entropy or publicly certifiable linear min-entropy, depending on the unforgeability guarantee of the one-shot signature scheme.

\begin{table}[t]
\centering
\footnotesize
\setlength{\tabcolsep}{3pt}
\renewcommand{\arraystretch}{1.45}
\begin{tabular}{@{}lcccc@{}}
\toprule
                           & \cite{NZ24}                                                   & \cite{BNZ24}       & \cite{CGS25}       & \textbf{This work} \\
\midrule
\rowcolor{gray!12} Model                      & {\renewcommand{\arraystretch}{1}\begin{tabular}{@{}c@{}}unitary\\quantum oracle\end{tabular}} & classical oracle   & classical oracle   & \emph{standard}    \\
Public-key / keyless       & keyless                & both               & public-key         & keyless            \\
\rowcolor{gray!12} Assumptions                & ---                    & {\renewcommand{\arraystretch}{1}\begin{tabular}{@{}c@{}}group\\pre-action\end{tabular}} & OSS & OSS \\
Provably untelegraphable   & \cmark                 & \xmark             & \cmark             & \cmark             \\
\rowcolor{gray!12} Conversability             & ?                      & ?                  & \xmark             & \cmark             \\
Unconversability           & ?                      & ?                  & \cmark             & \xmark             \\
\rowcolor{gray!12} Keyless untelegraphability & \xmark                 & \xmark             & \xmark             & \cmark             \\
\bottomrule
\end{tabular}
\caption{Comparison of known quantum fire constructions. \cmark{} = holds, \xmark{} = provably does not hold, ? = open, --- = none (unconditional); OSS = one-shot signatures. Rows are phrased so that \cmark{} is desirable. Conversability (transfer via interaction, used by our min-entropy application) and unconversability (a stronger guarantee ruling out such transfer) are mutually exclusive, but either may be desired, hence the two rows.}
\label{table:constructionComparison}
\end{table}

\subsection{Technical Overview}
\label{sec:tech_overview}

We now sketch our quantum-fire construction from one-shot signatures
and the main ideas behind its keyless untelegraphability and conversability
(\cref{thm:fire_standard_model,thm:informalConversability}). Before diving into the technical details, we recommend taking a quick look at \cref{app:visual_construction}, which gives a graphical presentation of the construction and may help build intuition.
In order to $\spark$ a flame, we generate a one-shot signature verification key $pk$ and quantum signing key $\qsk$, where the serial number is set to be $pk$, and the fire state is $\qsk$. Intuitively, this seems problematic. Recall that in a one-shot signature scheme, it should be computationally hard to find two signed messages that verify with the same verification key. Since a single copy of the quantum signing key allows signing any message, two information-theoretic clones would break the security of the one-shot signature scheme. In other words, it must be computationally hard to create information-theoretic clones of $\qsk$. The key observation is that quantum-fire cloning requires only functional cloning: the resulting states need not be two copies of $\qsk$; they only need to pass verification with the same serial number. Thus, given a flame $\qsk$ with serial number $pk$, the cloner creates two fresh OSS key pairs $(pk_0,\qsk_0)$ and $(pk_1,\qsk_1)$.
Then it runs $\sigma\gets \oneshotsign(\qsk,(pk_0,pk_1))$. The clones have a classical register that consists of the signed message for both clones (i.e., $(pk_0,pk_1,\sigma)$); the first clone's quantum register is $\qsk_0$ and the second's is $\qsk_1$. The classical register of each flame contains the authenticated root-to-leaf chain leading to its current verification key. Each link consists of two child verification keys, along with a signature over that pair using the selected verification key from the preceding link. The quantum register consists of a quantum signing key associated with one of the two verification keys in the last link.
More generally, repeated applications of the $\clone$ algorithm therefore form an authenticated binary tree of $\oss$ verification keys, rooted at the original serial number $pk$. Each cloned flame corresponds to a root-to-leaf path in this tree together with the quantum signing token associated with the leaf. 
Overall, the verification algorithm checks that:
\begin{enumerate}
    \item The signed message in each link passes verification with one of the two verification keys in the preceding link. (For the chain's head, the serial number functions as the only valid verification key.)
    \item The quantum register is a valid one-shot signing key relative to one of the two verification keys in the tail of the chain.
\end{enumerate}

We have two variants of the construction: in the logarithmic variant, formalized in \cref{const:fireFromOneShot}(a), along any root-to-leaf path, $\clone$  can be repeated at most $c\log(\secpar)$ times, for any fixed $c$ chosen in advance, and in the linear variant (see \cref{const:fireFromOneShot}(b)), at most $\frac{\secpar}{8}$ times.
We argue that the logarithmic scheme is keyless untelegraphable (see \cref{def:StrongTelegraphingSecurity}), based on the unforgeability of the OSS. As a warm-up, we first consider a weaker question: Can an adversary create two valid one-shot signing keys associated with the same verification key? The answer is no, since the two keys could be used to sign two different messages; these signed messages would be accepted with the same verification key. 
Suppose an adversary $\adv$ can generate a serial number $s$ and a classical message $t$ so that $\bdv$ can recover a valid flame with respect to $s$. Suppose, for simplicity, that the success probability of the adversaries is $1$, and that all the chains that are generated are exactly of length $c\cdot\log(\secpar)$ (the proof can easily be extended to the general case. See \cref{rem:nonPerfectAdvs}). Consider an adversary $\mathcal{F}$ that runs $(s,t) \gets \adv$ and then runs $\bdv(t)$ $\secpar^c+1$ times. Each such run creates a signature chain starting at $s$. We interpret that chain as part of a tree\footnote{Technically, this may not be a tree (and may have loops). The argument given below holds for the BFS tree generated from the graph described here; see the proof for the full details.}: if we see a valid signature on $(pk_0,pk_1)$ under $pk$, we attach $pk_0$ and $pk_1$ as $pk$'s children. Unless an OSS forgery occurs, the out-degree cannot exceed two: signed messages form pairs of verification keys, and having more than two children means that we have two valid signed messages with the same OSS verification key (i.e., an OSS forgery). Since the length of the chain is at most $c\cdot\log(\secpar)$, the binary tree has at most $\secpar^c$ leaves. Since $\bdv$ generated $\secpar^c+1$ such chains, by the pigeonhole principle, there must be two chains that terminate at the same leaf. As such, we have two one-shot signing keys, each associated with the same verification key as that leaf. As discussed in the previous paragraph, this results in a one-shot signature forgery and should therefore not be possible.

The linear scheme is keyless untelegraphable by the same argument, assuming
that the underlying OSS is exponentially unforgeable. In this case the cloning depth is $\secpar/8$, so the reduction may need to generate roughly $2^{\secpar/8}$ chains. Consequently, the resulting OSS forger runs in time
$\widetilde O(2^{\secpar/8})$, which is ruled out by sufficiently strong
exponential unforgeability.

Despite being keyless untelegraphable, the construction is conversable: Alice can send a fire state to Bob via interactive classical communication, using the following protocol. Bob generates a quantum one-shot signature key pair $(pk_B,\qsk_B) \gets \oneshotgen(\secparam)$, and sends $pk_B$ to Alice; Alice also generates a new key pair $(pk_A,\qsk_A) \gets \oneshotgen(\secparam)$. She uses the quantum register of the fire state (in the form of a one-shot signing key) to generate a signature $\sigma\gets \oneshotsign(\qsk,(pk_A,pk_B))$, appends the signed message to the chain, and sends the serial number $s$ and the resulting chain to Bob. At this point, both Alice and Bob have a quantum signing key and a chain of signatures from the flame's serial number to the verification key corresponding to their quantum signing key. Crucially, Bob generates the quantum state $\qsk_B$ locally before Alice responds. Alice never sends a quantum state; she only authenticates $pk_B$ as a descendant of her current flame. Thus, interaction enables Bob to obtain a valid flame even though one-way classical telegraphing remains impossible. The protocol can be repeated sequentially for logarithmically many rounds in the logarithmic variant, and for up to $\secpar/8$ rounds in the linear variant.

We next explain an important implication of keyless untelegraphability: the serial number of any efficiently generated valid flame must have super-logarithmic min-entropy. To see this, consider a QPT adversary that generates $(s,\quant{f})\gets \adv(\secparam)$, and assume for simplicity that the adversary always generates states that pass verification. 
Suppose that the min-entropy of the random variable $s$ is merely logarithmic, and let $s^*$ be the most probable outcome. The logarithmic min-entropy means that $s^*$ occurs with probability $\Omega(\secpar^{-c})$ for some constant $c$. 
Then, if Alice and Bob both run $\adv$, the probability that both generate valid flames with the serial number $s^*$ is at least $\Omega(\secpar^{-2c})$, contradicting keyless untelegraphability. 
(Note that Alice did not send Bob even a single bit, even though it is allowed.) For the linear construction (see \cref{const:fireFromOneShot}(b)), the argument is similar, but based directly on the exponential unforgeability guarantees of the OSS. These two results are shown formally in \cref{clm:Min-entropyFire,thm:min-entropy-linear-fire}.
We note that a variation of this idea was used by Zhandry in the context of quantum lightning, which is further discussed in the related-work section (\cref{sec:related_works}).

We now use these properties to build publicly certifiable min-entropy (defined in \cref{sec:publicCertifaiableMinEntropy}) from quantum fire. The construction is natural, given the properties of the conversable quantum fire that were discussed above: Alice generates a flame and uses its serial number as the min-entropy sample; to certify it, Alice and Bob run the $\conv$ protocol, at the end of which Bob obtains and verifies a valid flame with the same serial number. Bob can now act as the prover in a further execution, which is precisely the transferability property discussed above. The guarantees rest on the fire's properties: keyless untelegraphability ensures that valid serial numbers have high min-entropy, while clonability and conversability allow flames, and hence certifications, to spread using classical communication alone. Since our quantum fire is built from one-shot signatures, this yields \cref{thm:informalResultPubliclyCertifiableMinEntropy}; formally, it follows from the min-entropy implication above together with \cref{thm:ConversableFireToPublicallyCertifiableSecurity} in \cref{sec:PubliclyCertifiableMinEntropyFromconversable}.

So far, certification requires a quantum verifier. We show that for a decomposable conversable fire scheme, and assuming the OSS has classically generable verification keys, the verifier can instead remain classical. Informally, decomposability separates the classical authenticated chain from the quantum signing key, allowing the classical part of the conversability and verification procedures to be executed without generating the quantum token.

In our construction, the verifier can therefore generate only a classical verification key and verify the authenticated chain classically. The resulting protocol preserves the min-entropy guarantee but loses transferability: because the verifier does not possess the corresponding quantum signing key, it cannot subsequently serve as a prover for another verifier.

\subsection{Related Work}
\label{sec:related_works}

\paragraph{One-Shot Signatures.}
Amos, Georgiou, Kiayias, and Zhandry~\cite{AGKZ20} defined the notion of one-shot signatures as a signature scheme where every secret key can sign at most one (arbitrary) message. They showed a construction relative to a classical oracle; unfortunately, there was a bug in the security proof. Shmueli and Zhandry \cite{SZ25} showed two constructions of one-shot signatures from non-collapsing hash functions, one relative to a classical oracle, and another, the first in the standard model, based on subexponential indistinguishability obfuscation (iO), subexponentially secure one-way functions, and LWE. Shmueli and Zhandry's oracle-relative construction is shown to be exponentially unforgeable relative to the corresponding classical oracle. These two constructions are the main starting point of our work. Çakan, Goyal, and Shmueli~\cite{CGS25} then defined the property of incompressibility of one-shot signatures, a property regarding the inability to generate polynomially many signatures from a polynomial-sized string. One of the main drawbacks of incompressibility is that it applies in the oracle model and has no natural counterpart in the standard model. They further claimed that the oracle-based construction of \cite{SZ25} is incompressible, and used this to construct quantum fire (see the paragraph below); the claim was later shown to fail, and was repaired by Huang, Shmueli, Vaikuntanathan, and Zhandry~\cite{HSVZ26} via a modified oracle-model construction, recovering the results of \cite{CGS25} with perfect correctness.
They also improved the parameters of these schemes, allowing shorter keys and the signing of messages of length linear in the security parameter.

Our $\conv$ protocol is closely related to the signature-chain delegation
mechanism used by Amos, Georgiou, Kiayias, and Zhandry~\cite{AGKZ20} in
their blockchain-less cryptocurrency construction. There, authority is
transferred by signing a recipient's freshly generated verification key.
Our construction modifies this mechanism by signing pairs of fresh
verification keys, thereby turning the authenticated chain into a branching
structure that supports cloning. We use this structure, together with the
bounded depth of the chains, to obtain keyless untelegraphability and
conversability.

\paragraph{Quantum Fire.}
Nehoran and Zhandry \cite{NZ24} established an oracle separation between two significant quantum no-go theorems: the no-cloning and no-telegraphing theorems. Their result sheds light on the possibility of quantum states that are clonable yet untelegraphable. Building on this insight, Bostanci, Nehoran, and Zhandry~\cite{BNZ24} formalized the notion of quantum fire and proposed a candidate construction based on a new group-action hardness assumption; however, the question of whether this construction is untelegraphable remained open.

Çakan, Goyal, and Shmueli~\cite{CGS25}, as stated above, defined a new property of OSS, called incompressibility; they showed that relative to an oracle, the construction of \cite{SZ25} is incompressible, and also gave an untelegraphable quantum fire construction relative to that oracle. 
One drawback of their approach is its heavy reliance on the incompressibility of one-shot signatures. Since incompressibility is only defined in the oracle model, we are not aware of an approach for a standard-model incompressible OSS (in particular, the heuristic approach in which the oracles are obfuscated would \emph{not} be an incompressible one-shot signature). The same drawback prevents a natural standard-model construction of fire based on their approach. 

By contrast, our quantum fire constructions are in the standard model (see \cref{thm:fire_standard_model}). A comparative drawback of our constructions is that the 
number of allowed sequential clone operations is bounded. 
Note, however, that a clone operation creates two clones, and therefore, the overall number of clone operations and the number of resulting clones can nevertheless be exponential in our linear-clonable scheme (and polynomial in our logarithmic-clonable scheme).

\paragraph{Certified Randomness.}
Most relevant to our work is \emph{public} certifiable min-entropy, which has already been constructed in prior works.
First, Ref.~\cite[Section 7]{AGKZ20} shows that unforgeable one-shot signatures imply publicly certifiable min-entropy, relative to a classical oracle. This is done in two steps. In the first step, a 3-round \emph{private} certifiable min-entropy protocol is shown, in which the prover uses the one-shot signature token to sign a random message drawn by the verifier. In the second step, since the protocol is a public-coin interactive protocol, it is converted to a non-interactive \emph{public} protocol via the post-quantum Fiat--Shamir transform~\cite{FS86,LZ19,DFM20}. The proof transcript can then be distributed for anyone to verify, enabling \emph{public} certifiability of the verification key's entropy for the one-shot signature. The original construction in that work is in the classical oracle model, and the Fiat--Shamir transform additionally requires a random oracle.

Another line of work was initiated by Aaronson \cite{Aar18b,Aar18c,Aar20b} and later was formalized by Aaronson and Hung~\cite{AH23}. This line of work studies publicly certifiable randomness from random circuit sampling (RCS). Aaronson and Hung show that, under the Long List Quantum Supremacy Verification (LLQSV) hardness assumption, passing the linear cross-entropy benchmark implies $\Omega(n)$ min-entropy. Unlike the one-shot signatures line of work, their reduction to certifiable min-entropy does not rely on a random oracle (though it does provide random-oracle evidence for the LLQSV assumption). A major drawback of their approach is the need for an exponential-time classical verifier.
Inspired by Aaronson's work, Bassirian, Bouland, Fefferman, Gunn, and Tal~\cite{BBFGT26} study Fourier sampling in the quantum random oracle model (QROM). They show directly that any efficient quantum prover that passes the suggested Fourier-sampling test with non-negligible probability must generate strings from a distribution with at least linear min-entropy. This yields a single-round publicly certifiable linear min-entropy scheme. Their result has unconditional black-box security in the QROM, but, similar to Aaronson and Hung's work \cite{AH23}, it requires an exponential-time classical verifier that makes an exponential number of queries to the random oracle.

Finally, Yamakawa and Zhandry~\cite{YZ24} construct publicly verifiable proofs of quantumness relative to a random oracle. Assuming the Aaronson--Ambainis conjecture~\cite{AA14}, they establish a connection between the proofs of quantumness and the min-entropy of the distribution of the proof of their algorithm. Informally, they show that if the distribution of the algorithm's output has low min-entropy, then it is sufficiently concentrated around a few particular accepting proofs. The Aaronson--Ambainis conjecture can then be used to classically approximate the relevant output probabilities of a quantum prover with low min-entropy, thereby yielding a classical prover that contradicts the proof of quantumness.

Among the approaches above, Amos et al.~\cite{AGKZ20} and Yamakawa and
Zhandry~\cite{YZ24} rely on a random oracle to obtain public certifiability,
while the construction of Bassirian et al.~\cite{BBFGT26} is in the QROM.
Aaronson and Hung~\cite{AH23}, on the other hand, avoid a random oracle,
but require an exponential-time classical verifier.

In contrast, our construction has polynomial-time verification, and its
security reduction is in the standard model. Under ordinary unforgeability
of the underlying one-shot signature scheme, we obtain high min-entropy.
Under the stronger assumption that the one-shot signature scheme is
exponentially unforgeable, our construction yields linear min-entropy.
Such exponential unforgeability is known for the oracle-relative
construction of Shmueli and Zhandry~\cite{SZ25}, but, to the best of our
knowledge, is not currently known for any standard-model one-shot signature
scheme. Thus, our linear min-entropy result gives a standard-model reduction
from exponentially unforgeable one-shot signatures to publicly certifiable
linear min-entropy with polynomial-time verification.

Very recently, Khurana, Roberts, and Tal \cite{KRT26} revisited the Yamakawa--Zhandry \cite{YZ24} protocol and proved its certifiable min-entropy guarantee unconditionally (without relying on the Aaronson-- Ambainis conjecture) against low-depth quantum adversaries. More precisely, their result applies to adversaries making $o(\log(\secpar))$ rounds of queries to the random oracle, while allowing polynomially many parallel queries within each round. In contrast, our construction is in the standard model, while theirs remains in the random oracle model, and for non-shallow adversaries, the Aaronson-Ambainis conjecture is still required. 

Our construction uses quantum fire to obtain publicly certifiable
min-entropy with polynomial-time verification and transferability over
classical communication.

Note that a min-entropy source need not be uniformly random. The more ambitious goal is publicly certifiable randomness, in which the source is guaranteed to be $\epsilon$-close to uniform. 
There are two main approaches to private certified randomness, also known as randomness expansion, for a classical verifier. The first is device-independent randomness expansion, which is based on Bell's inequality. Here, it is assumed that two (potentially malicious) non-signaling quantum devices are spatially separated. The classical verifier initiates the protocol with a short, private random seed, and the goal is to generate a longer private string that is $\epsilon$-close to uniform~\cite{CR12,MS16,CY14}, or to terminate upon detecting that the devices misbehave.

One drawback of the previously mentioned results is that they require two distant, non-signaling devices. A complementary approach, introduced in~\cite{BCM+18}, is to use a single untrusted quantum device, whose output randomness is guaranteed, under a computational hardness assumption, against computationally bounded adversaries.

A key issue that arises in both approaches is \emph{accumulation}: these protocols consist of many iterations of a sub-protocol that generates a small number of bits of min-entropy (possibly less than one). The central argument is to demonstrate that this min-entropy accumulates when the sub-protocol is repeated appropriately. We note that, in this work, we do not show such an accumulation result, as we do not use the sub-protocol approach and instead argue directly about the resulting protocol. We leave to future work the question of whether some form of min-entropy accumulation can be achieved.

Another key difference between the private and public settings discussed in this work is \emph{privacy}. In the private setting: a) the random seed must be kept secret, and b) even though the verifier's string is random, the adversary has no knowledge of the resulting random string; one example where this is crucial is when this randomness is used for cryptographic private key generation. For publicly certifiable min-entropy and publicly certifiable randomness, the randomness is intentionally \emph{not} private; hence, property b) is clearly not satisfied.

\paragraph{Quantum Lightning.}
Zhandry~\cite{Zha21} defined quantum lightning, which is a collision-free version of public-key quantum money. From the uniqueness property of quantum lightning, one can obtain a private source of certified min-entropy from the distribution of the bolt's serial numbers~\cite[Theorem 3.3]{Zha21}. While quantum lightning provides \emph{private} certified min-entropy, our approach provides \emph{publicly certifiable} min-entropy that remains \emph{transferable} over classical communication channels.

\subsection{Open Problems}
We briefly state a few open problems.

\paragraph{Publicly Certifiable Randomness.}
Our certification protocols can publicly certify that a string was generated by an algorithm whose output distribution has high min-entropy. Yet, they stop short of certifying uniform randomness.
Can the public certifiability of \emph{min-entropy} be strengthened to a notion closer to uniform randomness? As an intermediate question, is there a way to achieve entropy accumulation?

\paragraph{Unbounded Clonability.}
As noted, even our strongest scheme allows at most a linear number of sequential $\clone$ operations (see \cref{thm:exponential-oss-imply-linear-fire}). As shown in \cref{prop:bad_oss}, this restriction is unavoidable in our construction. Are there constructions of conversable quantum fire in the standard model where the number of sequential $\clone$ operations can be any polynomial? We note that the quantum fire construction in Ref.~\cite{CGS25}, which is relative to a classical oracle, allows any number of $\clone$ operations.

\section{Definitions, Notation, and Preliminaries}
\label{sec:definitions}
We use calligraphic fonts with a capital letter to represent quantum algorithms (e.g., $\oneshotsign$), capitalized sans-serif for classical algorithms (e.g., $\oneshotversign$), lowercase calligraphic font for quantum states (e.g., $\qsk$ for a quantum signing key), and lowercase letters for classical variables (e.g., $m$). For cryptographic schemes, we use all capitals (e.g., $OSS$ for a one-shot signature scheme). To avoid ambiguity, we sometimes use the scheme followed by the algorithm's name (e.g., $OSS.\oneshotsign(\cdot)$).

In all our definitions below, we assume that all algorithms receive the security parameter $\secparam$, but we sometimes omit it to avoid clutter. 

\begin{definition} [Non-$\bot$ Min-Entropy] \label{def:nonbotMinEntropy}
    For a probability distribution $D=(p_\bot,p_1,\ldots,p_n)$, we define the non-$\bot$ min-entropy as:
    \begin{equation*}
        \Hbot(D) := \min_{i\neq \bot} \log(1/p_i).
    \end{equation*}
    In most cases, the probability $p_{\bot}$ will be clear from the context.
\end{definition}
\begin{remark}
    High min-entropy implies that with high probability, the sampled Kolmogorov complexity is high as well: $\Hinf(X)\geq e$ implies that $K(X)\geq e-d-O(1)$ with probability at least $1- 2^{-d}$, for any $d$. We will not interpret or further optimize our results through the lens of Kolmogorov complexity in this work.
    \label{rem:Kolmogorov}
\end{remark}

\subsection{One-Shot Signatures}
\label{sec:oss_def}

\begin{definition}[One-shot signatures, adapted from~\cite{AGKZ20}]\label{def:oneShotSignaturesSyntax}
    A one-shot signature scheme consists of four polynomial-time algorithms $(\oneshotgen,\oneshotsign,\oneshotversign,\oneshotvertoken)$ with the following syntax:
    \begin{itemize}
        \item $\oneshotgen(\secparam)$. It takes the security parameter as input and outputs a pair $(pk, \qsk)$, referred to as the verification key (public key) and the signing token (quantum secret key).
        \item $\oneshotsign(\qsk, m)$. It takes a quantum signing token, $\qsk$, and a message to be signed, $m$, and outputs a classical signature, $sig$.
        \item $\oneshotversign(pk, m, sig)$ is a classical deterministic polynomial-time algorithm. It takes the verification key, a message to be verified, and a signature, and outputs $\top$ or $\bot$. 
        \item $\oneshotvertoken(pk,\qsk)$ is a QPT algorithm that takes a verification key and a quantum signing key and outputs a quantum secret key $\qsk'$, and a bit $b$ indicating acceptance or rejection.
    \end{itemize}
\end{definition}
Some one-shot signature constructions use a \emph{common reference string} (CRS). When instantiating the construction in this work with such constructions, our schemes will also require a setup algorithm and a CRS, but we omit the syntax for that to avoid clutter. 

\begin{remark}[Implementing $\oneshotvertoken$]
The token-verification algorithm $\oneshotvertoken$ can be implemented using
$\oneshotsign$ and $\oneshotversign$ as follows. On input $(pk,\qsk)$, prepare
a message register in the maximally mixed state over the message space, and
run a coherent implementation of $\oneshotsign(\qsk,m)$, deferring all
measurements. Next, coherently compute
$\oneshotversign(pk,m,\sigma)$ into an output register. We then uncompute the
verification circuit and apply the inverse of the coherent signing circuit,
thereby undoing the signing computation and recovering the token register.
Finally, we measure only the output register and use the resulting bit as the
output of $\oneshotvertoken$.

In particular, the acceptance probability of $\oneshotvertoken(pk,\qsk)$ is
exactly the probability that, for a uniformly random message $m$, signing
$m$ using $\qsk$ produces a signature that passes
$\oneshotversign(pk,m,\cdot)$.
\end{remark}

\begin{definition}[One-shot signatures completeness, adapted from \cite{AGKZ20}]\label{def:oneShotCorrectness}
For all messages $m$ in the message space $\mathcal{M}$,
\begin{equation*}
    \Pr\left[\oneshotversign(pk, m, \oneshotsign(\qsk, m)) = \top:\bigstack{(pk,\qsk)\gets\oneshotgen(\secparam)}\right] = 1,
\end{equation*}

and 

\begin{equation*}
    \Pr\left[\oneshotvertoken(\oneshotgen(\secparam)) = \top\right] = 1.
\end{equation*}

\end{definition}

\begin{definition}[One-shot unforgeability, adapted from~\cite{AGKZ20}]
    Given a one-shot signature scheme $OSS = (\oneshotgen,\oneshotsign,\oneshotversign,\oneshotvertoken)$, we say $OSS$ is unforgeable if for all QPT adversaries $\adv$:

    \begin{equation*}
        \Pr\left[m_1\neq m_2 \wedge a = b = \top : \bigstack{(pk,m_1,sig_1,m_2,sig_2)\gets\adv(\secparam) \\ a\gets\oneshotversign(pk,m_1,sig_1) \\ b\gets\oneshotversign(pk,m_2,sig_2)}\right] \leq \negl.
    \end{equation*}

    We say that the scheme is \emph{exponentially} $h$-unforgeable if for every adversary $\adv$ that runs in time $O(2^{h'\secpar})$ for $h'<h$, the inequality above holds. 
    \label{def:OneShotUnforgeability}
\end{definition}

\subsection{Quantum Fire}
\label{sec:fire}

\begin{definition}[Adapted from~\cite{BNZ24}] \label{def:BNZfireSyntax}
    A quantum fire scheme consists of three QPT algorithms $(\spark,\clone,\ver)$ where:

\begin{itemize}
    \item $\spark(\secparam)$ takes the security parameter and outputs a serial number $s$ and a quantum fire state $\quant{f}^s$, which we refer to as a flame.
    \item $\clone(s, \quant{f}^s)$ takes a serial number $s$ and a flame $\quant{f}^s$, and outputs two registers $A$ and $B$ in some potentially entangled state $\quant{c}^s_{AB}$.
    \item $\ver(s,\quant{f}^s)$ takes a serial number $s$, and an alleged flame, and accepts or rejects. In $c$-clonable schemes, which will be introduced later (see \cref{def:CloningCorrectnessAlt}), we also output an integer $CL$ (Clones Left), which states the number of $\clone$ operations that can still be executed on the flame.
\end{itemize}
\end{definition}
We also define an algorithm $\tildever$ that runs $\ver(s,\quant{f}^s)$ and outputs $s$ if $\ver$ accepts, and $\bot$ otherwise.

Note that there can be multiple variants of quantum fire schemes~\cite{BNZ24,CGS25}, similar in spirit to private quantum money, public-key quantum money, and quantum lightning. In the definition above, we use the keyless variant (similar to quantum lightning, which has no private key).  

\begin{definition}[Correctness \cite{BNZ24}] \label{def:fireCorrectness}
\begin{equation*}
    \Pr\left[\ver(s,\quant{f}^s) = \top : (s, \quant{f}^s) \leftarrow \spark(\secparam) \right] = 1
\end{equation*}
\end{definition}

\begin{definition}[Cloning Correctness] \label{def:CloningCorrectnessAlt} We say that a quantum fire scheme is $c$-clonable if
\begin{equation*}
        \Pr\left[\ver(s,\quant{f}) =\top :
        \quant{f} \leftarrow \clone^{c(\secpar)}(\spark(\secparam)) \right] = 1,
    \end{equation*}
where $\clone^{c(\secpar)}$ is the algorithm that applies $\clone$ sequentially $c(\secpar)$ times, each time to one of the resulting registers.
We say that the scheme is one-time clonable if the above holds for $c(\secpar)=1$, as defined in~\cite{BNZ24}. 
\end{definition}

Next, we define two security notions for quantum fire: untelegraphability and keyless untelegraphability. We note that the original notion of untelegraphability, introduced in~\cite{BNZ24}, is not used in this work. Instead, we use a stronger notion, which we call \emph{keyless} untelegraphability, which is defined immediately below.  
\begin{definition}[Untelegraphability~\cite{BNZ24}] \label{def:BNM_untelegraphability}
We say the scheme is an untelegraphable quantum-fire scheme if for all (non-entangled) QPT adversaries $\adv$ and $\bdv$,
\begin{equation*}
    \Pr\left[\text{ Adversaries win }\right] \leq \negl,
\end{equation*}
where the untelegraphability game is defined as follows:

\pseudocodeblock{
\textbf{C} \<\< \textbf{A} \<\< \textbf{B}
\\[0.1\baselineskip][\hline] \\[-0.5\baselineskip]
(s, \quant{f}^s) \leftarrow \spark(\secparam) \\
\< \sendmessageright*{s,\quant{f}^s} \\
\<\<\dots \\
\<\<\< \sendmessageright*{t \in \{0,1\}^*} \\
\<\<\<\< \dots \\
\< \sendmessageleftx{8}{\quant{g}} \\
b \gets \ver(s, \quant{g}) \\
}

The adversaries win if $b=\top$.
\end{definition}
We define a stronger variant:
\begin{definition}[Keyless untelegraphability] \label{def:StrongTelegraphingSecurity}
    We use the same definition as untelegraphability (\cref{def:BNM_untelegraphability}), except we change the security game as follows:
    \pseudocodeblock{
        \textbf{C} \<\< \textbf{A} \<\< \textbf{B}
        \\[0.1\baselineskip][\hline] \\[-0.5\baselineskip]
          \< \sendmessageright*{\secparam} \\
        \<\<\dots \\
        \< \sendmessageleft*{s}\<\< \sendmessageright*{t \in \{0,1\}^*} \\
        \<\<\<\< \dots \\
        \< \sendmessageleftx{8}{\quant{f}} \\
        b\gets \ver(s,\quant{f}) \\
    }
\end{definition} 
\cref{def:StrongTelegraphingSecurity} implies \cref{def:BNM_untelegraphability}, since $\adv$ could always choose to generate $s$ honestly by running the $\spark$ algorithm. 

\begin{definition}[Min-Entropy of Quantum Fire] \label{def:minEntropyFire}
We say that a fire scheme has $e(\secpar)$ min-entropy if for every QPT adversary $\mathcal{L}$,
\[ \Hbot(\tildever(\mathcal{L}(\secparam))) \geq e(\secpar). \]
(Recall \cref{def:nonbotMinEntropy} for $\Hbot$, and \cref{def:BNZfireSyntax} for $\tildever$.) We say that a fire scheme has high min-entropy if it has $c \cdot \log (\secpar)$ min-entropy for every constant $c>0$. We say it has linear min-entropy if there exists a constant $c>0$ such that the scheme has $c \secpar$ min-entropy.\footnote{Note that linear min-entropy implies high min-entropy, and not necessarily vice versa.}
\end{definition}

Quantum fire is defined to be untelegraphable, meaning a flame cannot be sent via one-way classical communication. 
Conversable fire, a notion we introduce in this work, \emph{can} be sent via two-way (1-round) classical communication:
\begin{definition}[Conversability] \label{def:Conversability}
A conversable quantum fire is a quantum fire scheme with three additional QPT algorithms $\prepare$, $\decon$, $\recon$. 
\begin{itemize}
    \item $\prepare(1^\lambda)$ takes the security parameter and outputs a classical preparation string and a quantum preparation state.
    \item $\decon(s,\quant{f},p)$ takes a serial number $s$, a flame $\quant{f}$, and the preparation string $p$ (a classical string). The algorithm outputs a classical telegraph string and a verification key.
    \item $\recon(s,t,vk,\quant{g})$ takes the serial number $s$, the telegraph string $t$, a verification key $vk$, and the quantum preparation state $\quant{g}$, and outputs a flame.
\end{itemize}

The $\conv$ protocol is a protocol over classical communication between a sender, who wants to send a flame, and a receiver; it is defined as follows:
\begin{enumerate}
    \item The receiver runs $(r,\quant{g})\gets\prepare(\secparam)$, and sends $r$ to the sender.
    \item The sender starts with a pair of a serial number and a flame, $(s,\quant{f})$, runs $d,vk \gets \decon(s,\quant{f},r)$, and sends $d,vk$ to the receiver.
    \item The receiver runs $\quant{f'} \gets \recon(s,d,vk,\quant{g})$.
\end{enumerate}

We say the $\conv$ protocol succeeds if $\recon$ does not output $\bot$.

For $c$-clonable schemes, let $(s,\quant{f}^s)$ be the result of $\spark(\secparam)$ followed by at most $c(\secpar)$ combinations of $\clone$ operations and $\conv$ protocol executions. We say that the scheme is conversable if:
\begin{equation*}
    \Pr\left[\ver(s,\quant{f}') = \top :\bigstack{
(r,\quant{g})\gets\prepare(\secparam) \\ 
(d,vk)\gets\decon(s,\quant{f},r) \\
\quant{f}'\gets\recon(s,d,vk,\quant{g})}\right] 
= 1.
\end{equation*}
    
\end{definition}

The definition can be easily extended to $n$-round conversability to allow multiple rounds of interactive communication, but this is not needed in our work.

\subsection{Publicly Certifiable Min-Entropy over Classical Communication}
\label{sec:publicCertifaiableMinEntropy}
\begin{definition} [Publicly Certifiable Min-Entropy]  \label{def:publicCertifiableSyntax}
    A publicly certifiable min-entropy scheme consists of three QPT algorithms ($\certgenent$, $\certclonewitness$, $\certver$) where:
    \begin{itemize}
        \item $\certgenent(\secparam)$ takes the security parameter and generates a string $s$ and a quantum witness $\quant{w}_s$.
        \item $\certclonewitness(s, \quant{w}_s)$ takes a string, $s$ and a quantum witness, $\quant{w}_s$, and outputs a quantum state with two registers $\quant{w}_1, \quant{w}_2$, each of which is an individual quantum witness. 
        \item $\certver\langle P, V\rangle$ is a polynomial-time protocol, over classical communication, between a quantum prover and a quantum verifier, where the prover takes a string and a quantum witness $\quant{w}$, and the verifier's input is the security parameter. The verifier ends with a classical bit representing acceptance or rejection, a classical string $s'$, and a quantum witness $\quant{w}'$.
    \end{itemize}
\end{definition}

Similarly to quantum fire (see \cref{def:BNZfireSyntax}), we define an algorithm $\widetilde{\certver}$ that runs $\certver$ and outputs $s$ if the verifier accepts or $\bot$ otherwise.

\begin{definition} [$c$-clonability] \label{def:CloningWitnessCorrectness}
    We say a publicly certifiable min-entropy scheme is $c$-clonable for a function $c$ if the following holds:
    \begin{equation*}
        \Pr\left[\certver(s,\quant{w}) = \top  :\quant{w}\gets\certclonewitness^{c(\secpar)}(\certgenent(\secparam))\right] = 1,
    \end{equation*}
    where $\certclonewitness^{c(\secpar)}$ is the algorithm that applies $\certclonewitness$ sequentially at most $c(\secpar)$ times, each time keeping either the first or the second output register.
\end{definition}

\begin{definition} [Public Certifiability Correctness] \label{def:PublicCertCorrectness}
A publicly certifiable min-entropy scheme is $c(\secpar)$-correct if 
    \begin{equation*}
        \Pr\left[\certver(s,\quant{w}_P) = \top : \quant{w}_P\otimes\quant{w}_V\gets \certver^{c(\secpar)}(\certgenent(\secparam))\right] = 1,
    \end{equation*}
    where $\certver^{c(\secpar)}$ is the protocol that runs the $\certver$ protocol sequentially for $c(\secpar)$ repetitions. 
\end{definition}

\begin{definition} [High Min-Entropy Certifiability] \label{def:publicCertifiableMin-Entropy} 
    We say that a publicly certifiable min-entropy scheme has $e(\secpar)$ min-entropy if for every QPT adversary $\mathcal P$,
\begin{equation*}
        \Hbot\left[ \widetilde{\certver}\langle \mathcal P, V\rangle\right] \geq e(\secpar).
\end{equation*}

    We say that a publicly certifiable min-entropy scheme has high min-entropy if, for every constant $c$, the scheme has $c \cdot \log(\secpar)$ min-entropy. We say that a scheme has linear min-entropy if there exists a constant $c$ such that it has $c\cdot \secpar$ min-entropy.

\end{definition}

We can now define the notion of transferability:

\begin{definition} [$c$-Transferability] \label{def:C-Transferability}
We say that a min-entropy scheme is $c$-transferable for a function $c=c(\secpar)$ if: 

    \begin{equation*}
        \Pr\left[\widetilde{\certver}(s,\quant{w}) \neq \bot :  (s, \quant{w})\gets \certver^{c(\secpar)}(\certgenent(\secparam))\right] = 1,
    \end{equation*}
where $\certver^{c(\secpar)}$ is a sequence of $c(\secpar)$ runs of the (honest) $\certver$ protocol, in which, after each run the output state becomes the prover's witness for the next round.
\end{definition}

\section{Conversable Quantum Fire from One-Shot Signatures}
\label{sec:ConversableFireFromOneShot}

The main result in this section is the construction of conversable quantum fire from one-shot signatures (for intuition, we recommend first taking a quick look at \cref{app:visual_construction}, which provides a graphical presentation of the construction), together with its security proof. We start by defining the two variants of the construction:
\begin{construction}[Conversable Quantum Fire from One-shot Signatures] \label{const:fireFromOneShot}
    Given a one-shot signature scheme $OSS = (\oneshotgen,\oneshotsign,\oneshotversign,\oneshotvertoken)$ and a function $c(\secpar)$, we define a $c(\secpar)$-clonable conversable quantum fire scheme. We use the term logarithmic construction, or \cref{const:fireFromOneShot}(a), for the construction below with $c(\secpar)=c \cdot \log(\secpar)$ for any constant $c$; and the linear construction, or \cref{const:fireFromOneShot}(b), with $c(\secpar)=\frac{1}{8}\secpar$. (The main role of $c$ is line \ref{it:check_length} in $\ver$.) 
    \begin{itemize}
        \item $\spark(\secparam)$:
        \begin{enumerate}
            \item $(pk,\qsk) \gets \oneshotgen(\secparam)$
            \item Output $(s:=pk, \quant{f}^s:=\emptyset,\qsk)$ \footnote{The classical register of the flame is an empty string in the beginning.}
        \end{enumerate}
        \item $\clone(s,\qsk)$:
        \begin{enumerate}
            \item Let $\qsk'$ be the quantum register in $\qsk$, and let $sk''$ be the first classical register. 
            \item $(\qsk_{0},pk_{0},\qsk_{1},pk_{1}) \gets \oneshotgen(\secparam)^{\tensor 2} $.
            \item $sig \gets \oneshotsign(\qsk',(pk_0||pk_1))$.
            \item Set $\qsk_0:=((sk'',pk_0,pk_1,sig),pk_0,\qsk_0)$.
            \item Set $\qsk_1:=((sk'',pk_0,pk_1,sig),pk_1,\qsk_1)$.
            \item Output $\qsk_0,\qsk_1$.
        \end{enumerate}
        \item $\ver(s, \qsk)$: 
         \begin{enumerate}
            \item Interpret the first classical register in $\qsk$ as $(pk^1_0,\allowbreak pk^1_1,\allowbreak sig_1,\allowbreak pk^2_0,\allowbreak pk^2_1,\allowbreak sig_2,\allowbreak \ldots,\allowbreak pk^T_0,\allowbreak pk^T_1,\allowbreak sig_T)$, the second classical register as $pk$, and the quantum part $\qsk'$; and denote $pk_0^0:=s$, and $pk_1^0:=s$. 
            \item \label{it:check_length} Reject if $T>c(\secpar)$.
            \item For every $t$ from $1$ to $T$:
            \begin{enumerate}
                \item $a_t \gets \oneshotversign_{pk^{t-1}_{0}}(pk^{t}_0||pk^t_1,sig_t)\vee \oneshotversign_{pk^{t-1}_{1}}(pk^{t}_0||pk^t_1,sig_t)$. 
            \end{enumerate}
            \item Reject if $(pk^{T}_{0} \neq pk \wedge pk^{T}_{1} \neq pk)$.
            \item $a,\widetilde\qsk\gets \oneshotvertoken(pk,\qsk')$ \label{it:vertoken}
            \item $CL = c(\secpar) - T$. \label{it:CL_calc}
            \item Output $a \wedge\left( \bigwedge_{t=1}^T a_t\right), CL$. \label{it:bigand}
        \end{enumerate}
        \item $\prepare(\secparam)$:
        \begin{enumerate}
            \item Output $\oneshotgen(\secparam)$.
        \end{enumerate}
        \item $\decon(s,\qsk,r)$:
        \begin{enumerate}
            \item Interpret $\qsk$ as $ch,pk,\qsk'$.
            \item $(\widetilde{pk},\qsk'')\gets\oneshotgen(\secparam)$.
            \item $sig\gets\oneshotsign(\qsk',\widetilde{pk} || r)$.
            \item $ch':= (ch,\widetilde{pk},r, sig)$. (Note: this chain is one link longer than the one we started with.) 
            \item Output $d:=ch',vk:=r$.
        \end{enumerate}
       \item $\recon(s,d,vk,\quant{sk}')$:
            \begin{enumerate}
                \item Interpret $d$ as $ch'$, and let $vk'$ be the second verification key in the last link of $ch'$.
                \item If $vk\neq vk'$, reject. Otherwise, output
                $(ch',vk,\quant{sk}')$.
            \end{enumerate}
    \end{itemize}
\end{construction}

Next, we analyze the case in which the clonability is logarithmically bounded (\cref{const:fireFromOneShot}(a)).
\begin{theorem} \label{thm:ConversableShortVariant}
    For any constant $c$, \cref{const:fireFromOneShot}(a) is a keyless untelegraphable (see~\cref{def:StrongTelegraphingSecurity}) conversable $c \cdot \log(\secpar)$-clonable fire scheme if $OSS$ is an unforgeable one-shot signature scheme (see \cref{def:OneShotUnforgeability}).
\end{theorem}
\begin{proof}
    The correctness (\cref{def:fireCorrectness}), conversability (\cref{def:Conversability}), and the $c \log(\secpar)$-clonability (\cref{def:CloningCorrectnessAlt}) follow from the correctness of the one-shot signature scheme $OSS$.
     Assume, towards a contradiction, that the scheme is not keyless untelegraphable, i.e., there exist two adversaries $\adv$ and $\bdv$ that succeed in the security game defined in \cref{def:StrongTelegraphingSecurity}.  For simplicity, we assume their success probability is 1 (the general case, where the success probability is non-negligible, can be shown similarly using a simple repetition argument). We define the following forger, $\mathcal{F}$:
     \begin{enumerate}
         \item $(s,t)\gets\adv(1^\secpar)$.
         \item For each $i$ from $1$ to $\lceil 2^{c(\secpar)}\rceil+1$ run
         $\qsk_i\gets\bdv(t)$.
         \item For each pair $(i,j)$:
         \begin{enumerate}
             \item \label{it:inconsistent_branches} \textbf{Inconsistency test:} Check the classical registers of $\qsk_i$ and $\qsk_j$ to see if the two branches of signed messages are inconsistent, that is, whether we can find two different messages signed relative to the same $pk$; if so, return the $pk$ and the two messages and signatures.
            \item \label{it:branch_and_oss_leaf} \textbf{Internal node test:} Denote by $pk_i$ and $pk_j$ the associated verification keys of $\qsk_i$ and $\qsk_j$. If the $i$th branch contains a signature that passes verification with $pk_j$, then use $\qsk_j$ to sign a different message, and return $pk_j$ and the two signed messages.
             \item \label{it:two_leaves}\textbf{Two terminals test:} If $pk_i=pk_j$, $sig_0 \gets \oneshotsign(\qsk_i,0)$, $sig_1 \gets \oneshotsign(\qsk_j,1)$. Return $(pk_i,0,sig_0,1,sig_1)$.
        \end{enumerate}
        \item \label{it:output_bot} Return $\bot$.
     \end{enumerate}
     Since $\adv$ and $\bdv$ are QPT, $\mathcal F$ is also a QPT algorithm. We argue that $\mathcal{F}$ always succeeds in outputting a forgery.
     We use the $\lceil 2^{c\cdot \log(\secpar)}\rceil+1$ chains to construct a directed graph $G$, as follows: The source of the graph is the node labeled $s$. Each $pk$ is a vertex, and we add an edge from a vertex $pk$ to $pk_0$ and $pk_1$ if there is a signed message of $(pk_0,pk_1)$ that passes verification with respect to $pk$. Even though it may be useful to view this graph as a tree (and, when generated honestly, it \emph{is} a tree), it need not be. 
     Clearly, if we output two signed messages in step~\ref{it:inconsistent_branches}, we have a successful forger in the sense of one-shot unforgeability (\cref{def:OneShotUnforgeability}). A failure to find such a case means that the out-degree of all vertices is at most 2. It remains to consider steps~\ref{it:branch_and_oss_leaf} and \ref{it:two_leaves}, in which $\mathcal F$ finds a
     live signing token associated with a verification key $pk$. By the
     implementation of $\oneshotvertoken$ described above, the probability
     that such a token passes $\oneshotvertoken$ is exactly the probability
     that, on a uniformly random message $m$, the signature
     $\sigma\gets\oneshotsign(\qsk,m)$ satisfies
     $\oneshotversign(pk,m,\sigma)=\top$.

     Therefore, whenever step~\ref{it:branch_and_oss_leaf} is reached, $\mathcal F$ chooses a uniformly
     random message $m'$ different from the message already signed under $pk$
     and signs $m'$ using the live token. With non-negligible probability this
     signature verifies, in which case $\mathcal F$ has obtained two distinct
     valid signed messages under the same verification key. Similarly, in
     step~\ref{it:two_leaves}, $\mathcal F$ chooses two distinct uniformly random messages and
     signs one with each of the two live tokens associated with the same
     verification key. With non-negligible probability both signatures verify,
     again yielding two distinct valid signed messages under the same
     verification key. In either case, $\mathcal F$ obtains an OSS forgery. Hence, whenever either step is reached, $\mathcal F$ outputs a valid OSS forgery.
     Lastly, we have to argue why the forger never reaches step~\ref{it:output_bot} (and hence never outputs $\bot$). 
    Recall that in step~\ref{it:check_length}, we make sure that the length of the chain of signed messages is at most $c\log(\secpar)$. Therefore, the distance between $s$ and every other vertex in the graph is at most $c\log(\secpar)$. A graph with out-degree at most $2$, in which every vertex is at distance at most $d$ from $s$, has at most $2^d$ terminals. This can be proved by induction, or by noticing that a breadth-first search (BFS) traversal on $G$ returns a binary spanning tree for $G$ with height at most $d$, and therefore the BFS tree has at most $2^d$ leaves. Every terminal in $G$ must be a leaf in the BFS tree, and therefore the number of terminals in $G$ is at most $2^d$. By the pigeonhole principle, since we have $\lceil 2^{c\cdot \log(\secpar)}\rceil +1$ such chains, we either have two terminals associated with the same verification key, and would therefore return a non-$\bot$ output in step~\ref{it:two_leaves}; or a one-shot signing key associated with an internal node, which yields a non-$\bot$ value in step~\ref{it:branch_and_oss_leaf}. 
\end{proof}

\begin{remark}
The assumption above that the adversaries succeed with probability $1$ is only for simplicity.
Suppose instead that their success probability is some non-negligible
$\varepsilon(\lambda)$. By averaging over the output $(s,t)$ of $\mathcal A$,
with non-negligible probability over $(s,t)\leftarrow\mathcal A(1^\lambda)$,
the conditional probability that $\mathcal B(t)$ outputs a valid flame with
serial number $s$ is itself non-negligible. Conditioned on such an output,
we run $\mathcal B(t)$ independently a polynomial factor more times.
By a Chernoff bound, except with negligible probability, more than
$2^{c(\lambda)}$ of these executions output valid flames. We can then apply
the same pigeonhole argument as in the proof above to these valid flames.
This is only an analysis of the successful executions; the forger need
not run $\ver$ to identify them before applying the pigeonhole argument.

The additional polynomial factor does not affect the polynomial running time in
the logarithmic variant, nor does it affect the exponential exponent of the
running time in the linear variant.
\label{rem:nonPerfectAdvs}
\end{remark}

For \cref{const:fireFromOneShot}(a), where $c(\secpar) = c\cdot\log(\secpar)$, we only need a standard level of unforgeability, i.e., security against polynomially bounded adversaries, in order to prove our construction is keyless untelegraphable. As we will now see, for \cref{const:fireFromOneShot}(b) where $c(\secpar) = \frac{1}{8}\secpar$, we need unforgeability even against exponential-time adversaries:

\begin{theorem}
    If $OSS$ is an exponentially $h$-unforgeable one-shot signature with $h>\frac{1}{8}$ (see \cref{def:OneShotUnforgeability}), then the linear construction, \cref{const:fireFromOneShot}(b), is a keyless untelegraphable $\frac{\secpar}{8}$-clonable conversable quantum fire scheme. 
    \label{thm:exponential-oss-imply-linear-fire}
\end{theorem}

\begin{proof}
    The proof is very similar to the proof of \cref{thm:ConversableShortVariant}.
    As in the previous proof, correctness, conversability, and $c(\lambda)$-clonability follow from the correctness of one-shot signatures.

    Assume, towards a contradiction, that there exist adversaries $\adv$ and $\bdv$ that win quantum fire's keyless untelegraphability game. For simplicity, assume first that the adversaries succeed with probability $1$; the case of general non-negligible success probability is discussed at the end of the proof.

    Let $\mathcal F$ be the forger from the proof of \cref{thm:ConversableShortVariant}, with $c(\secpar)=\secpar/8$. Thus, $\mathcal F$ generates $N:=\lceil 2^{\secpar/8}\rceil+1$ flames by running $\bdv$. We implement the search in Step 3 of that forger efficiently, rather than explicitly examining every pair of generated flames.

    For each generated flame, $\mathcal F$ parses its classical authenticated chain and its terminal verification key. For every signed link appearing in one of these chains, $\mathcal F$ stores a record $(pk,m,\sigma)$ whenever $\oneshotversign_{pk}(m,\sigma)=\top$. Let $\mathcal R$ be the resulting list of records. Since there are $N$ chains, each of length at most $\secpar/8$, the total number of records in $\mathcal R$ is at most $N\cdot \secpar/8$.

    The forger sorts $\mathcal R$ according to the verification key $pk$, and also sorts the list of terminal verification keys $pk_1,\ldots,pk_N$. It then performs the following three tests:
    \begin{enumerate}
        \item \emph{Inconsistency test.} If $\mathcal R$ contains two records $(pk,m_0,\sigma_0)$ and $(pk,m_1,\sigma_1)$ with $m_0\neq m_1$, then $\mathcal F$ outputs $(pk,m_0,\sigma_0,m_1,\sigma_1)$.

        \item \emph{Internal node test.} If some terminal verification key $pk_i$ also appears as the first component of a record $(pk_i,m,\sigma)\in\mathcal R$, then $\mathcal F$ chooses any message $m'\neq m$, computes $\sigma'\gets\oneshotsign(\qsk_i,m')$, and outputs $(pk_i,m,\sigma,m',\sigma')$.

        \item \emph{Two terminals test.} If there exist $i\neq j$ such that $pk_i=pk_j$, then $\mathcal F$ chooses two distinct messages $m_0,m_1$, computes $\sigma_0\gets\oneshotsign(\qsk_i,m_0)$ and $\sigma_1\gets\oneshotsign(\qsk_j,m_1)$, and outputs $(pk_i,m_0,\sigma_0,m_1,\sigma_1)$.
    \end{enumerate}

    These are exactly the three cases considered in the proof of \cref{thm:ConversableShortVariant}. If the inconsistency test does not succeed, every vertex in the graph defined there has out-degree at most $2$. If the internal node test does not succeed, all $pk_1,\ldots,pk_N$ are terminal vertices of the graph, and if the two terminals test does not succeed, these terminal vertices are all distinct. Since every valid chain has length at most $\secpar/8$, every vertex is at distance at most $\secpar/8$ from the source. As shown in the proof of \cref{thm:ConversableShortVariant}, such a graph has at most $2^{\secpar/8}$ terminal vertices. This contradicts the fact that $\mathcal F$ generated $N=\lceil 2^{\secpar/8}\rceil+1$ distinct terminal vertices. Hence one of the three tests must succeed, and $\mathcal F$ outputs a valid OSS forgery.

    It remains to bound the running time of $\mathcal F$. The number of signed-link records is at most $N\cdot \secpar/8=O(2^{\secpar/8}\cdot\poly(\secpar))$. The three tests above can be implemented by sorting the signed-link records and the terminal verification keys. Since sorting introduces only a logarithmic overhead, and all one-shot signature verification operations take polynomial time, the total running time of $\mathcal F$ is $O(2^{\secpar/8}\cdot\poly(\secpar))$.

    Since $h>1/8$, fix any constant $h'$ such that $1/8<h'<h$. For sufficiently large $\secpar$, we have $2^{\secpar/8}\cdot\poly(\secpar)=O(2^{h'\secpar})$. Therefore, $\mathcal F$ runs within the time bound ruled out by exponential $h$-unforgeability and succeeds with probability $1$, yielding a contradiction.

    Finally, the case in which $\adv$ and $\bdv$ succeed with general non-negligible probability follows from the amplification argument in the remark following \cref{thm:ConversableShortVariant}. That argument introduces only an additional polynomial factor in the number of executions of $\bdv$, and therefore the resulting forger still runs in time $O(2^{\secpar/8}\cdot\poly(\secpar))$. Hence the same contradiction applies.
\end{proof}

From the last two theorems, we obtain the following corollary:

\begin{corollary}
    An unforgeable one-shot signature scheme implies the logarithmic
    variant of keyless untelegraphable, conversable quantum fire.
    If the one-shot signature scheme is exponentially $h$-unforgeable for some $h>\frac18$, then it also implies the linear variant of keyless untelegraphable, conversable quantum fire.
    \label{coro:OssToConversableFire}
\end{corollary}

In both variants of our construction, the chain length is bounded by a function, either logarithmic in the security parameter or $\secpar/8$ (see line \ref{it:check_length} in $\ver$). A natural question is whether the variant in which this condition is relaxed or removed completely is also secure. We give a negative answer to this question by showing that security cannot hold generically. (Moreover, the result below shows that even a bound equal to the length of the verification key cannot work generically.) 

\begin{proposition}
    Let $OSS$ be a one-shot signature scheme. There exists a black-box construction (see~\cref{const:bad_oss}) of a scheme $OSS'$ from $OSS$ such that $OSS'$ is unforgeable if and only if $OSS$ is. 
    
    Furthermore, \cref{const:fireFromOneShot} instantiated with $c(\secpar)=\secpar$ is not keyless untelegraphable when constructed using $OSS'$.
    \label{prop:bad_oss}
\end{proposition}

\begin{proof}
    We start by defining $OSS'$ as follows:

    \begin{construction}
    Let $OSS = (\oneshotgen,\oneshotsign,\oneshotversign,\oneshotvertoken)$ be a one-shot signature scheme where we assume without loss of generality that the verification keys are of length $\secpar$. We define $OSS'$ to be a one-shot signature scheme such that the algorithms $\oneshotgen$, $\oneshotsign$, $\oneshotvertoken$ are the same, the alphabet of the verification keys includes a new symbol, "\$", and the $\oneshotversign(vk,m,\sigma)$ algorithm is as follows:
    \begin{enumerate}
        \item If the verification key has the form $\$||vk$:
        \begin{enumerate}
            \item If $|vk| < \secpar-1$ and $m=(\$||vk || 0, \$||vk || 1)$, then accept.
            \item If $|vk| = \secpar-1$ and $m=(vk || 0, vk || 1)$, then accept.
            \item Else reject.
        \end{enumerate}
        \item Else output $OSS.\oneshotversign(vk,m,sig)$.
    \end{enumerate}
    \label{const:bad_oss}
\end{construction}
    
    We show that if $OSS$ is an unforgeable one-shot signature scheme, then $OSS'$ is also unforgeable. First, correctness holds since we did not change the $\oneshotgen$ algorithm. It still generates a token satisfying correctness with a verification key over the original alphabet (see \cref{def:oneShotCorrectness}); hence $OSS'.\oneshotversign$ acts as $OSS.\oneshotversign$ on such keys. 
    Next, unforgeability (see \cref{def:OneShotUnforgeability}) holds. Consider an adversary that forges two signed messages, $m_0\neq m_1$, associated with the same verification key. In the case where the verification key does not begin with "\$", $OSS'.\oneshotversign$ acts as $OSS.\oneshotversign$, and therefore the probability of the adversary producing such a forgery is negligible. Note that every $vk$ that begins with "\$" will only be accepted with a unique message as defined in \cref{const:bad_oss}. Therefore, finding two distinct messages that pass verification under such a key is impossible.
    Conversely, any forgery against $OSS$ is also a forgery against $OSS'$,
    since on verification keys over the original alphabet,
    $OSS'.\oneshotversign$ acts identically to $OSS.\oneshotversign$.

    With that in hand, we now prove that \cref{const:fireFromOneShot}, constructed from $OSS'$, is telegraphable. We define the following adversaries $\adv$ and $\bdv$ for the keyless untelegraphability game (see \cref{def:StrongTelegraphingSecurity}). Perhaps surprisingly, our adversaries do not send \emph{any} communication to each other. $\adv$ sends the serial number $\$$ to the challenger, and sends an empty string to $\bdv$. $\bdv$ runs $(vk,\qsk)\gets\oneshotgen(\secparam)$ and generates a classical register that starts with $\$0,\$1$. Let $vk_{:i}$ be the first $i$ bits of $vk$. For each bit of $vk$,
    with index $i$, $\mathcal B$ appends the following string to the end of
    the register: $\$vk_{:i}0,\$vk_{:i}1$.
    For every such link, $\mathcal B$ appends an arbitrary dummy signature
    $\sigma_\bot$, since $OSS'.\oneshotversign$ ignores the signature whenever
    the verification key begins with "$\$$".
    For the last bit of $vk$, the appended string does not begin with the
    "$\$$" symbol. $\bdv$ ends up with the classical register containing the string 
    \begin{equation*}
    (\$0,\$1,\sigma_\bot),
    (\$vk_{:1}0,\$vk_{:1}1,\sigma_\bot),
    \ldots,
    (vk_{:\secpar-1}0,vk_{:\secpar-1}1,\sigma_\bot).
    \end{equation*}
    $\bdv$ then sends a quantum state whose first classical register
    contains the string constructed above, whose second classical register is
    $vk$, and whose quantum register contains $\qsk$, the token generated by
    $\oneshotgen$. This string passes verification under \cref{const:fireFromOneShot} since $OSS'.\oneshotversign$ accepts the entire chain by the added verification rules. This means that these two adversaries win the keyless untelegraphability game with probability 1. 
\end{proof}

We now examine the min-entropy of \cref{const:fireFromOneShot}. We first show that any keyless untelegraphable quantum fire scheme has high min-entropy.

\begin{theorem}[Quantum Fire Min-Entropy] \label{clm:Min-entropyFire}
     Keyless untelegraphability of quantum fire (see \cref{def:StrongTelegraphingSecurity}) implies high min-entropy, as defined in \cref{def:minEntropyFire}.
\end{theorem}
\begin{proof} 
    Assume towards a contradiction that there exists an adversary $\mathcal{L}$ that outputs $(s',\quant{f}')$ whose distribution has low min-entropy (i.e., violates the guarantee in \cref{def:minEntropyFire}). We define the following two adversaries to quantum fire's keyless untelegraphability game: $\adv$ runs $(s',\quant{f}')\gets\mathcal{L}$ and sends $s'$ to the challenger. $\bdv$ runs $(s'',\quant{f}'')\gets\mathcal{L}$, and sends $\quant{f}''$ to the challenger. Let $s^*$ be the most probable serial number generated by $\mathcal{L}$, which occurs with probability $p^*$, where $p^* = 2^{-\Hbot(\widetilde\ver(\mathcal{L}(\secparam)))}$. Because we assume $\mathcal{L}$ generates serial numbers and flames with low min-entropy, we can say that $p^*$ is non-negligible. The probability that $\adv$ and $\bdv$ win the untelegraphability game is at least the probability that both obtain the same serial number from $\mathcal{L}$, which is at least $(p^*)^2$, and hence non-negligible. 
\end{proof}

Note that no communication between the adversaries is needed, even though $\adv$ is allowed to send a classical message to $\bdv$.

Conversable quantum fire is quantum fire with an additional correctness notion of conversability and thus also has the same notion of min-entropy. We now show that the same exponential-unforgeability assumption used for \cref{const:fireFromOneShot}(b) also yields min-entropy linear in the security parameter.

\begin{theorem} 
    If $OSS$ is an exponentially $\frac{1}{3}$-unforgeable one-shot signature scheme (see \cref{def:OneShotUnforgeability}), then \cref{const:fireFromOneShot} with $c(\secpar)=\frac{1}{8}\secpar$ has min-entropy at least $\frac{1}{8}\secpar$ (see \cref{def:minEntropyFire}).
    \label{thm:min-entropy-linear-fire}
\end{theorem}

\begin{proof}
    Assume towards a contradiction that some adversary $\mathcal{L}$ satisfies $\Hbot(\tildever(\mathcal{L}(\secparam)))$ $< \frac{1}{8}\secpar$. We define the forger $\mathcal{F}$ to be the following algorithm: 
    $\mathcal{F}$ runs $\mathcal{L}$ $2^{\frac{1}{4}\cdot\secpar + 2}$ times. For every serial number $s$ that $\mathcal{L}$ generated, define the graph $G(s)$ to be the directed graph defined in the proof of \cref{thm:ConversableShortVariant} from the branches generated by $\mathcal{L}$ that start with $s$ and pass verification. If any of these graphs has an inconsistency, then $\mathcal{F}$ finds a forgery. Otherwise, let $s^*$ be the most common serial number $\mathcal{L}$ generates. The probability that $\mathcal{L}$ outputs $s^*$ is $p^* > 2^{-\frac{1}{8}\cdot\secpar}$. The expected number of branches that are part of $G(s^*)$ is then at least $2^{\frac{1}{4}\cdot\secpar + 2} \cdot 2^{-\frac{1}{8}\cdot\secpar} = 2^{\frac{1}{8}\cdot\secpar + 2}$. By the Chernoff--Hoeffding inequality, with probability at least $\frac{1}{2}$, there are at least $2^{\frac{1}{8}\secpar +1}$ chains starting with $s^*$. 

    As in the proof of \cref{thm:ConversableShortVariant}, since
    $G(s^*)$ contains no inconsistency, every vertex has an out-degree of at most $2$. Moreover, every vertex is at a distance of at most $c(\secpar)=\frac18\secpar$ from $s^*$. Therefore, by the BFS-tree argument from that proof, $G(s^*)$ has at most $2^{\frac18\secpar}$ terminal vertices.

    Now consider the at least $2^{\frac{1}{8}\secpar+1}$ chains generated starting at $s^*$. If one of these chains terminates at an internal
    vertex of $G(s^*)$, then, exactly as in the internal-node case of \cref{thm:ConversableShortVariant}, the forger obtains an OSS
    forgery. Otherwise, all of these chains terminate at terminal vertices. Since there are at most $2^{\frac18\secpar}$ such vertices,
    by the pigeonhole principle, two of the chains terminate at the same verification key. As in the two-terminal case of \cref{thm:ConversableShortVariant}, this yields an OSS forgery.
    
    This means that we have constructed a forger with a constant success probability, with runtime $O(2^{\frac{1}{4}\cdot\secpar + 2}\cdot\poly)$, which contradicts exponential unforgeability.
\end{proof}

Note that a more general result can be achieved by taking $c(\secpar)=\ell\cdot\secpar$ for some fixed $\ell \in (0,h)$, although there is a trade-off between clonability and min-entropy: if the min-entropy is $e\cdot\secpar$ for some fixed constant $e$, then $\ell+e < h$.

We now lay the groundwork for a classically certifiable min-entropy scheme (see \cref{const:classicalPCME}). We start by defining a variant of conversable fire, with a classical receiver in the $\conv$ protocol. 

\begin{definition} [Decomposable Fire]
    A conversable quantum fire scheme $CQFIRE$ is \emph{decomposable} if the flame can be decomposed into three registers, two classical and one quantum, and there exist three PPT algorithms 
    $\ifprepare(\secparam)$, $\ifrecon(r,s,d)$, and $\ifver(s,ch,vk)$, where:
    \begin{itemize}
        \item $\ifprepare$ takes the security parameter and outputs a verification key (classical string),
        \item $\ifrecon$ takes a serial number $s$ and a deconstructed fire string $d$ and outputs the classical registers $ch$ and $vk$, or $\bot$, and
        \item $\ifver$ takes a serial number, a chain string, and a verification key and outputs $\top$ or $\bot$,
    \end{itemize}
    that satisfy, for every (computationally unbounded) adversary $\adv$, the following:
        \begin{align*}
        \Pr&\left[\ifver(s, ch,vk) = \top:\bigstack{r \gets \ifprepare(\secparam) \\ s,d \gets \adv(r) \\ ch,vk \gets \ifrecon(r,s,d)}\right] \\
        =&\Pr\left[\ver(s,ch,vk, \quant g')= \top:\bigstack{r,\quant g \gets \prepare(\secparam) \\ s,d \gets \adv(r) \\ ch,vk,\quant g' \gets \recon(r,s,d,\quant g)}\right].    
        \end{align*}

    We also define a variant of the verification algorithm, $\widetilde{\ifver}(s, ch,vk,r)$, that, when $\ifver$ accepts, outputs the serial number $s$ and, when it rejects, outputs $\bot$.  
    \label{def:DecomposableFire}
\end{definition}

As with the other primitives we have discussed, we define a notion of min-entropy to accompany decomposable fire.

The next step in laying the groundwork is to show that the min-entropy is preserved when moving from a conversable quantum fire to a decomposable fire.

\begin{lemma}
    If $CQFIRE$ is a decomposable fire (see \cref{def:DecomposableFire}) and has $e(\secpar)$ min-entropy, then for every QPT adversary $\mathcal{CL}$:
    \begin{equation}
        \Hbot\left(
            \widetilde \ifver(s, ch,vk) = \top:\bigstack{r \gets \ifprepare(\secparam) \\ s,d \gets \mathcal{CL}(r) \\ ch,vk \gets \ifrecon(r,s,d)}
        \right) \geq e(\secpar)
        \label{eq:decomposableDefinition}
        \end{equation}
    \label{lemma:decomposableFirePreservesMinEntropy}
\end{lemma}

\begin{proof}
    Fix an adversary $\mathcal{CL}$. We define the following two distributions, for all $s\in\{0,1\}^* \cup \{\bot\}$:
    \begin{equation*}
        D_1(s) := \Pr\left[\widetilde \ifver(s, ch,vk) = s :\bigstack{r \gets \ifprepare(\secparam) \\ s,d \gets \mathcal{CL}(r) \\ ch,vk \gets \ifrecon(r,s,d)}\right]
    \end{equation*}
    \begin{equation*}
        D_2(s) := \Pr\left[\widetilde\ver(s, ch,vk,\quant{g}') = s:\bigstack{r,\quant{g} \gets \prepare(\secparam) \\ s,d \gets \mathcal{CL}(r) \\ ch,vk,\quant{g}' \gets \recon(r,s,d,\quant{g})}\right]
    \end{equation*}
    We argue that $D_1=D_2$, i.e., $\forall s \in\{0,1\}^* \cup \{\bot\} :D_1(s)=D_2(s)$. Assume towards a contradiction that there exists $s^*\in\{0,1\}^*$ such that $D_1(s^*) \neq D_2(s^*)$. We define the following adversary $\adv'(r)$ relative to \cref{def:DecomposableFire}. 
    $\adv'(r)$ runs $s,d \gets \mathcal{CL}(r)$; if $s=s^*$, it outputs $(s^*,d)$; otherwise, it outputs $\bot$. If we examine Eq.~\ref{eq:decomposableDefinition}, the left-hand side with our $\adv'$ is equal to $D_1(s^*)$, and the right-hand side is equal to $D_2(s^*)$. We have reached a contradiction. Since the two distributions are equal, their non-$\bot$ min-entropies are equal as well:
    \begin{align*}
        &\Hbot\left(
            \widetilde \ifver(s, ch,vk) = s :\bigstack{r \gets \ifprepare(\secparam) \\ s,d \gets \mathcal{CL}(r) \\ ch,vk \gets \ifrecon(r,s,d)}
        \right) \\
        &=\Hbot\left(
            \widetilde\ver(s, ch,vk,\quant{g}') = s:\bigstack{r,\quant{g} \gets \prepare(\secparam) \\ s,d \gets \mathcal{CL}(r) \\ ch,vk,\quant{g}' \gets \recon(r,s,d,\quant{g})}
        \right) \geq e(\secpar),
    \end{align*}
    where the final inequality comes from \cref{def:minEntropyFire}, taking the adversary $\mathcal{L}$ to be the composition of $\prepare$, $\mathcal{CL}$, and $\recon$ appearing in the equation above.
\end{proof} 

\begin{theorem}
    For every OSS that satisfies completeness (\cref{def:oneShotCorrectness}), if there exists a classical algorithm $\oneshotcgen$ such that $\oneshotcgen(\secparam)=\oneshotgen(\secparam)_{vk}$, then the fire in \cref{const:fireFromOneShot} is decomposable.
    \label{thm:ossToDecomposable}
\end{theorem}

\begin{proof}
    We want to prove that the decomposability conditions hold. We observe that the fire states generated in \cref{const:fireFromOneShot} can indeed be partitioned into three registers in the required format. We can now define the three algorithms: 
    \begin{enumerate}
        \item $\ifprepare(\secparam):= \oneshotcgen(\secparam)$
        \item We define $\ifver(s, ch, vk)$ to execute the first four steps of \cref{const:fireFromOneShot}'s $\ver$ algorithm and output the conjunction of all the $a_t$ values. We define $\ifqver(ch,\quant f)$ to be line \ref{it:vertoken} of $\ver$.  
        \item We define $\ifrecon(r,s,d)$ to act like $\recon$, while disregarding
        the quantum state. More precisely, it interprets $d$ as $ch'$ and lets
        $vk$ be the second verification key in the last link of $ch'$. It rejects
        if $vk\neq r$. Otherwise, it outputs $(ch',vk)$.
    \end{enumerate}

    We now show that the condition holds:
        
    \begin{align*}
        &\Pr_{\substack{r,\quant g \gets \prepare(\secparam) \\ s,d \gets \adv(r) \\ ch, vk, \quant g \gets \recon(s,d,r,\quant g)}}\left[\ver(s,(ch, vk,\quant g)) = \top \right]  \\
        \stackrel{(*)}{=}&\Pr_{\substack{r,\quant g \gets \oneshotgen(\secparam) \\ s,d \gets \adv(r) \\ ch, vk, \quant g \gets \recon(s,d,r,\quant g)}}\left[\ifver(s,ch,vk) \wedge \oneshotvertoken(vk,\quant{g}) = \top \right]  \\
        \stackrel{(**)}{=}&\Pr_{\substack{r,\quant g \gets \oneshotgen(\secparam) \\ s,d \gets \adv(r) \\ ch, vk, \quant g \gets \recon(s,d,r,\quant g)}}\left[\ifver(s,ch,vk) \wedge vk = r \wedge \oneshotvertoken(r,\quant{g}) = \top \right] \\
        =&\Pr_{\substack{r,\quant g \gets \oneshotgen(\secparam) \\ s,d \gets \adv(r) \\ ch, vk, \quant g \gets \recon(s,d,r,\quant g)}}\left[\ifver(s,ch,vk) \wedge \oneshotvertoken(r,\quant{g}) = \top \right] \\
        \stackrel{\text{\cref{def:oneShotCorrectness}}}{=}& \Pr_{\substack{r \gets \ifprepare(\secparam) \\ s,d \gets \adv(r) \\ ch,vk \gets \ifrecon(r,s,d)}}\left[\ifver(s,ch,vk) = \top \right] \\
        &\text{(*) follows by the definition of $\ver$.} \\
    &\text{(**) $\ifver$ checks that $vk = r$.}
    \end{align*}
    Since the condition is satisfied by our construction of $\ifprepare,\ifrecon$ and $\ifver$, there exist such algorithms as required by \cref{def:DecomposableFire}, and therefore \cref{const:fireFromOneShot} is indeed decomposable.
\end{proof}

The existence of such a classical algorithm $\oneshotcgen$ for the construction
of~\cite{SZ25} follows implicitly from~\cite[Proposition~29]{SZ25} in their
ePrint version.

Combining this fact with \cref{thm:ConversableShortVariant,thm:ossToDecomposable}
and the standard-model one-shot signature construction of~\cite{SZ25}, we obtain:

\begin{corollary}
    Under subexponentially secure iO,
    subexponentially secure OWF, and polynomially secure LWE with
    a subexponential noise-modulus ratio, \cref{const:fireFromOneShot}(a) instantiated with
    the one-shot signature scheme of~\cite{SZ25} is a keyless untelegraphable,
    conversable, decomposable quantum fire scheme.
\end{corollary}

\section{Publicly Certifiable Min-Entropy over Classical Communication from Conversable Quantum Fire}
\label{sec:PubliclyCertifiableMinEntropyFromconversable}

\begin{construction} \label{const:TPCME}
    Given a conversable quantum fire scheme $CQFIRE = ($$\spark$, $\clone$, $\ver$, $\langle\prepare,\decon,\recon\rangle)$, we define the following construction:

    \begin{itemize}
        \item $\certgenent(\secparam)$:
        \begin{enumerate}
            \item Output $\spark(\secparam)$.
        \end{enumerate}
        \item $\certclonewitness(s, \quant{w})$:
        \begin{enumerate}
            \item Output $\clone(s, \quant{w})$.
        \end{enumerate}
        \item $\certver\left\langle\certprover(\secparam,s,\quant{w}),\certcertifer(\secparam)\right\rangle$:
        \begin{enumerate}
            \item The $\certcertifer$ runs $(t,\quant{w}') \gets \prepare(\secparam)$, and sends $t$ to $\certprover$.
            \item The prover runs $d,vk \gets \decon(\secparam,s,\quant{w},t)$ and sends $s$, $d$ and $vk$ to $\certcertifer$.
            \item $\certcertifer$ runs $\quant{w}''\gets\recon(s,d,vk,\quant{w}')$
            \item $\certcertifer$ runs $b \gets \ver(s,\quant{w}'')$
            \item Output $b, s, \quant{w}''$
        \end{enumerate}
    \end{itemize}
\end{construction}

\begin{theorem} 
\label{thm:ConversableFireToPublicallyCertifiableSecurity}
    If $CQFIRE$ is a $c(\secpar)$-clonable conversable quantum fire with high
    min-entropy (see \cref{def:minEntropyFire}), then \cref{const:TPCME} is a publicly certifiable,
    $c(\secpar)$-transferable min-entropy scheme (see~\cref{def:publicCertifiableMin-Entropy}).
\end{theorem} 
\begin{proof}
    The correctness notion of $c$-clonability of witnesses follows directly from the cloning correctness of quantum fire. The notions of transferability and public certifiability follow from the conversability of the quantum fire scheme.

    It remains to prove high min-entropy. Assume, towards a contradiction, that \cref{const:TPCME} does not have high min-entropy. Then there exist a constant $c>0$ and a QPT adversary $\mathcal P$ such that, for infinitely many values of $\secpar$,
    \[
        \Hbot\left[\widetilde{\certver}\langle\mathcal P,V\rangle\right]
        < c\log(\secpar).
    \]

    We construct a QPT adversary $\mathcal L$ against the min-entropy of the underlying quantum fire scheme. The adversary $\mathcal L$ simulates the honest certifier in the $\certver$ protocol with $\mathcal P$. In particular, it runs
    $(t,\quant{w}')\gets\prepare(\secparam)$, sends $t$ to $\mathcal P$, receives $(s,d,vk)$ from $\mathcal P$, and computes
    $\quant{w}''\gets\recon(s,d,vk,\quant{w}')$.
    It then outputs $(s,\quant{w}'')$.

    By the definition of \cref{const:TPCME}, the distribution
    \[
        \widetilde{\ver}\bigl(\mathcal L(\secparam)\bigr)
    \]
    is identical to the distribution
    \[
        \widetilde{\certver}\langle\mathcal P,V\rangle.
    \]
    Therefore, for infinitely many values of $\secpar$,
    \[
        \Hbot\left[\widetilde{\ver}\bigl(\mathcal L(\secparam)\bigr)\right]
        < c\log(\secpar),
    \]
    contradicting the high min-entropy of the underlying quantum fire scheme as defined in \cref{def:minEntropyFire}.
\end{proof}

We can now prove Main Result 3, presented in \cref{thm:informalResultPubliclyCertifiableMinEntropy} using \cref{thm:ConversableShortVariant,clm:Min-entropyFire,thm:min-entropy-linear-fire,thm:ConversableFireToPublicallyCertifiableSecurity,thm:exponential-oss-imply-linear-fire}:

\begin{corollary}
    A conversable quantum fire constructed from an unforgeable one-shot
    signature scheme implies a publicly certifiable min-entropy scheme with
    at least super-logarithmic min-entropy that is $c\log(\secpar)$-transferable,
    for any fixed constant $c$. If the one-shot signature scheme is
    $\frac{1}{3}$-exponentially unforgeable, then the resulting scheme has at least linear
    min-entropy and is $\secpar/8$-transferable.
\label{coro:OssToTransferablePublicCertifiable}
\end{corollary}

An important question is whether we can eliminate the need for a quantum certifier. We answer this question in the affirmative, but this has some drawbacks. First, we require that our scheme be constructed from a decomposable conversable quantum fire (see \cref{def:DecomposableFire}). Second, a classical certifier does not receive a quantum witness at the end of the certification protocol, and therefore the resulting classically certifiable scheme is not transferable.

\begin{definition}[Classically certifiable min-entropy scheme]
    A classically certifiable min-entropy scheme is defined as in
    \cref{def:publicCertifiableSyntax}, except that the certifier in
    $\certver$ is a classical PPT algorithm. Consequently, the certifier outputs only an acceptance bit and a classical string $s'$ and is not required to output a quantum witness. We denote this new protocol between a QPT prover and a PPT verifier $\certcver$.

    Correctness is defined analogously to
    \cref{def:PublicCertCorrectness}, with $\certcver$ in place of
    $\certver$ and with no quantum output witness for the certifier.

    We also define $\widetilde{\certcver}$ to output $s'$ if the certifier accepts, and $\bot$ otherwise.
\label{def:classicalCertifiableSyntax}
\end{definition}

\begin{definition} [High Min-Entropy Classical Certifiability]
    We say that a classically certifiable min-entropy scheme has $e(\secpar)$ min-entropy if for every QPT adversary $\mathcal{P}$,
    \begin{equation*}
        \Hbot\left[\widetilde\certcver\langle\mathcal{P},V\rangle\right] \geq e(\secpar).
    \end{equation*}

    We say that a classically certifiable min-entropy scheme has high min-entropy if, for every constant $c$, the scheme has $c\cdot\log(\secpar)$ min-entropy. We say it has linear min-entropy if there exists a constant $c$ such that the scheme has $c\cdot\secpar$ min-entropy.
    \label{def:classicalCertifiability}
\end{definition}

\begin{construction}
    Let $CQFIRE$ be a decomposable quantum fire scheme (see \cref{def:DecomposableFire}).
    We construct a classically certifiable min-entropy scheme exactly as in \cref{const:TPCME}, except that $\certver$ is replaced by the following protocol $\certcver$:
    \begin{enumerate}
        \item The certifier runs $r\gets\ifprepare(\secparam)$ and sends $r$ to the prover.
        \item The prover runs $(d,vk)\gets\decon(s,\quant{w},r)$ and sends $s,d$ to the certifier.
        \item The certifier runs $(ch,vk')\gets\ifrecon(r,s,d)$.
        \item The certifier computes $b\gets\ifver(s,ch,vk')$.
        \item Output $b,s$.
    \end{enumerate}
    \label{const:classicalPCME}
\end{construction}

\begin{proposition}
    If $CQFIRE$ is a \emph{correct} decomposable conversable quantum fire scheme, then \cref{const:classicalPCME} is a \emph{correct} classically certifiable min-entropy scheme.
    \label{prop:classicalCertCorrectness}
\end{proposition}

\begin{proof}
    By decomposability of $CQFIRE$, $\ifprepare$, $\ifrecon$ and $\ifver$ are PPT classical algorithms, so the certifier in \cref{const:classicalPCME} is classical.

    It remains to prove correctness. Consider an honest execution of $\certcver$. By the defining condition of decomposable fire (see \cref{def:DecomposableFire}), the acceptance probability of the classical verification procedure, obtained by running $\ifprepare$, $\ifrecon$, and $\ifver$, is equal to the acceptance probability of the corresponding quantum procedure obtained by running $\prepare$, $\recon$, and $\ver$.

    The latter experiment is exactly the honest $\conv$ protocol used in \cref{const:TPCME} to implement $\certver$, and therefore accepts with probability $1$ by conversability (\cref{def:Conversability}). Hence $\certcver$ is correct.
\end{proof}
    
\begin{theorem}
    If $CQFIRE$ is a decomposable quantum fire scheme with $e(\secpar)$ min-entropy, then the classically certifiable min-entropy scheme given in \cref{const:classicalPCME} has $e(\secpar)$ min-entropy (see \cref{def:classicalCertifiability}). 
    \label{thm:CFireMinEntropyImpliesClassicallyCertifiablity}
\end{theorem}

\begin{proof}
    Fix any QPT prover $\mathcal P$. Define the QPT adversary $\mathcal{CL}$ from \cref{lemma:decomposableFirePreservesMinEntropy} by letting $\mathcal{CL}(r)$ run $\mathcal P$ on input $r$ and output the resulting pair $(s,d)$.

    By the definition of $\certcver$, the distribution of
    $\widetilde{\certcver}\langle\mathcal P,V\rangle$
    is exactly the distribution obtained by running
    $r\gets\ifprepare(\secparam)$,
    $(s,d)\gets\mathcal{CL}(r)$,
    $(ch,vk)\gets\ifrecon(r,s,d)$,
    and outputting $s$ if $\ifver(s,ch,vk)$ accepts, and $\bot$ otherwise.

    Therefore, by \cref{lemma:decomposableFirePreservesMinEntropy},
    \[
        \Hbot\left[\widetilde{\certcver}\langle\mathcal P,V\rangle\right]
        \geq e(\secpar).
    \]
\end{proof}

By combining \cref{thm:ConversableShortVariant,clm:Min-entropyFire,thm:ossToDecomposable,prop:classicalCertCorrectness,thm:CFireMinEntropyImpliesClassicallyCertifiablity} we can now state:

\begin{corollary}
    An unforgeable one-shot signature scheme with classically generable verification keys implies a classically certifiable min-entropy scheme with high min-entropy.
    \label{coro:OssToClassicallyCertifiable}
\end{corollary}

\paragraph{AI Disclosure.} We used Claude (Anthropic) and ChatGPT (OpenAI) only to assist with proofreading, routine editorial work, guidance on the literature review, and the generation of illustrations. These tools did not generate or materially affect any substantive content: all definitions, theorems, constructions, and proofs are the authors' own, every suggested edit was reviewed by the authors, and the authors take full responsibility for the entire content.

\ifnum\anonymous=0
    \ifnum\lipics=0
        \ifnum\masterthesis=0
            \subsection*{Acknowledgments}
    We wish to thank Zvika Brakerski for fruitful discussions and for suggesting the main ideas used in \cref{prop:bad_oss}, Taiga Hiroka for notifying us of an inaccuracy in \cref{def:PublicCertCorrectness} in an earlier version of this work, Mark Zhandry for several suggestions and corrections, and Yoav Elem for the assistance with the visuals. 
    This research was supported by the Israel Science Foundation (grant No. 2527/24).

            \BeforeBeginEnvironment{wrapfigure}{\setlength{\intextsep}{0pt}}
            \begin{wrapfigure}{r}{100px}
                \includegraphics[width=100px]{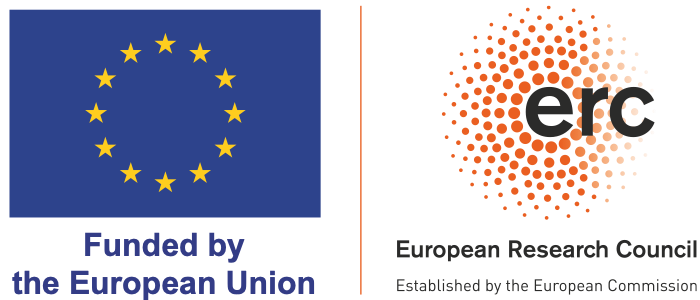}
            \end{wrapfigure}
            This work was funded by the European Union (ERC-2022-COG, ACQUA, 101087742). Views and opinions expressed are, however, those of the author(s) only and do not necessarily reflect those of the European Union or the European Research Council Executive Agency. Neither the European Union nor the granting authority can be held responsible for them.
        \fi
    \fi
\fi
\ifnum\sigconf=1
    \bibliographystyle{ACM-Reference-Format}
\else
    \ifnum\cryptology=1
        \bibliographystyle{abbrv}
    \else
        \ifnum\lipics=1
            \bibliographystyle{plainurl}
        \else
            \ifnum\llncs=1
                \bibliographystyle{splncs04}
            \else 
                \bibliographystyle{alphaabbrurldoieprint}
            \fi
        \fi
    \fi
\fi
\ifnum\masterthesis=0
    \ifnum\smallbib=1
        {\footnotesize \bibliography{main} }
    \else 
        \bibliography{main}
    \fi
\fi

\appendix
\ifnum\shownomenclature=1
    \printnomenclature[1in]
\fi

\section{A Visual Representation of One-Shot Signatures, Conversable Quantum Fire, and the Conversable Quantum Fire Construction from One-Shot Signatures}
\label{app:visual_construction}

The purpose of this appendix is to give an informal graphical presentation of \cref{const:fireFromOneShot} and the primitives that are involved. 

\subsection{Legend}
\label{sec:legend}

\Cref{fig:visual-legend} defines the graphical language used in the illustrations throughout this appendix.

\begin{figure}[htbp]
    \centering
    \input{figs/images/legend/legend-main}
    \caption{\textbf{Graphical notation used in this appendix.}
    Colored squares represent classical strings, such as public keys or serial numbers; different colors denote different classical strings.}
    \label{fig:visual-legend}
\end{figure}
\subsection{Algorithms and Security Games}

\Cref{fig:oss-syntax,fig:fire-syntax} show the syntax of both one-shot signatures (see definition: \cref{def:oneShotSignaturesSyntax}) and quantum fire (see definition: \cref{def:BNZfireSyntax}), displayed using the graphical language we have defined. We omit the security parameter $\secparam$ in these figures for clarity.

\begin{figure}[htbp]
    \centering
    \input{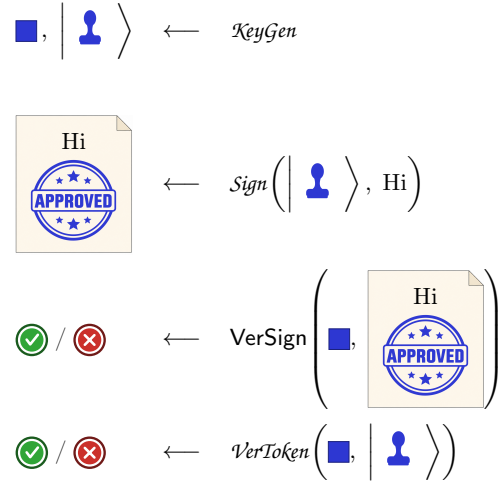}
    \caption{\textbf{One-shot-signature syntax.}
    Note that an honestly generated one-shot signature and public-key pair share the \emph{same color}.}
\label{fig:oss-syntax}
\end{figure}

\begin{figure}[htbp]
    \centering
    \input{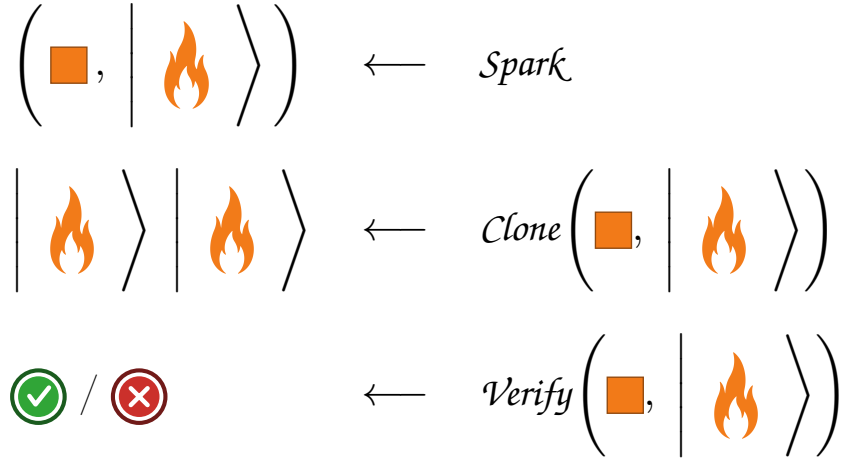}
    \caption{\textbf{Quantum-fire syntax.}
    Similar to one-shot signatures, an honestly generated flame and its corresponding serial number share the same color.}
\label{fig:fire-syntax}
\end{figure}

Next, \cref{fig:oss-unforgeability-visual,fig:keyless-untelegraphability-visual} give a graphical representation of the security definitions of one-shot signatures (unforgeability \cref{def:OneShotUnforgeability}) and quantum fire (keyless untelegraphability \cref{def:StrongTelegraphingSecurity}).

\begin{figure}[htbp]
    \centering
    \input{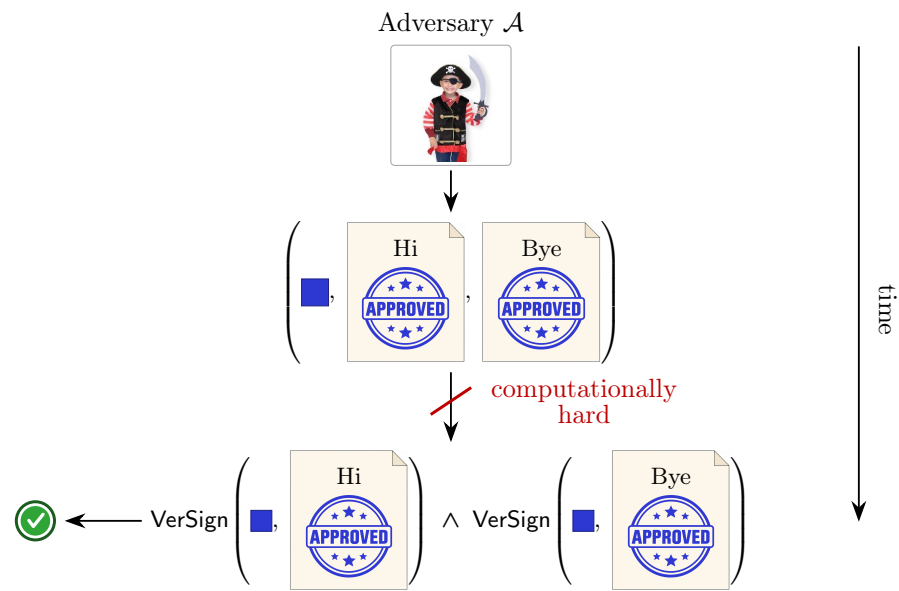}
    \caption{\textbf{One-shot unforgeability.}
    Even an adversary that chooses the verification key should not be able to produce two different signed messages that are accepted under the same verification key.}
\label{fig:oss-unforgeability-visual}
\end{figure}

\begin{figure}[htbp]
    \centering
    \input{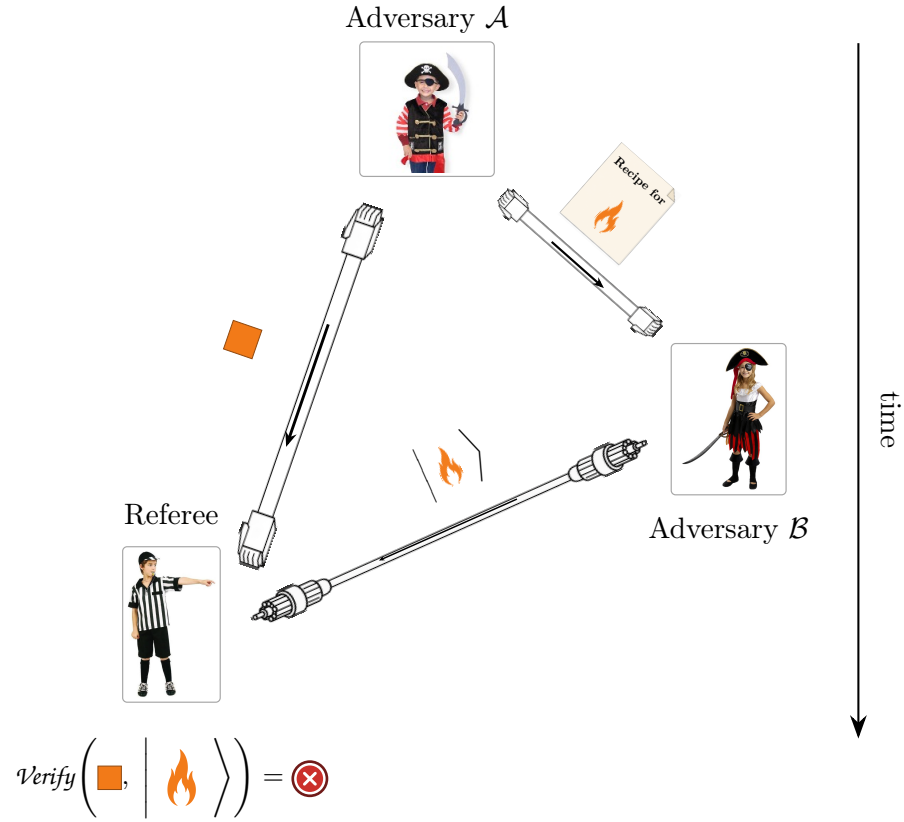}
    \caption{\textbf{Keyless untelegraphability.}
    Keyless untelegraphability requires that no efficient pair of adversaries can make the referee accept with non-negligible probability.}
\label{fig:keyless-untelegraphability-visual}
\end{figure}

\subsection{From one-shot signatures to conversable quantum fire}
\label{sec:oss_to_fire_visual}

We now show how we construct each one of conversable quantum fire's algorithms, and the $\conv$ protocol, from unforgeable one-shot signatures. A description of the construction is available in \cref{sec:tech_overview}, and the formal construction can be found in \cref{sec:ConversableFireFromOneShot}. \Cref{fig:spark-from-oss,fig:clone-visual,fig:verify-visual,fig:converse-visual} show $\spark$, $\clone$, $\ver$ and $\conv$, respectively.

\begin{figure}[htbp]
    \centering
    \input{figs/images/construction/spark/spark-main}
    \caption{\textbf{Sparking a new flame.} The blue public key of the one-shot signature takes the role of the serial number of the blue flame.}
    \label{fig:spark-from-oss}
\end{figure}

\begin{figure}[htbp]
    \centering
    \input{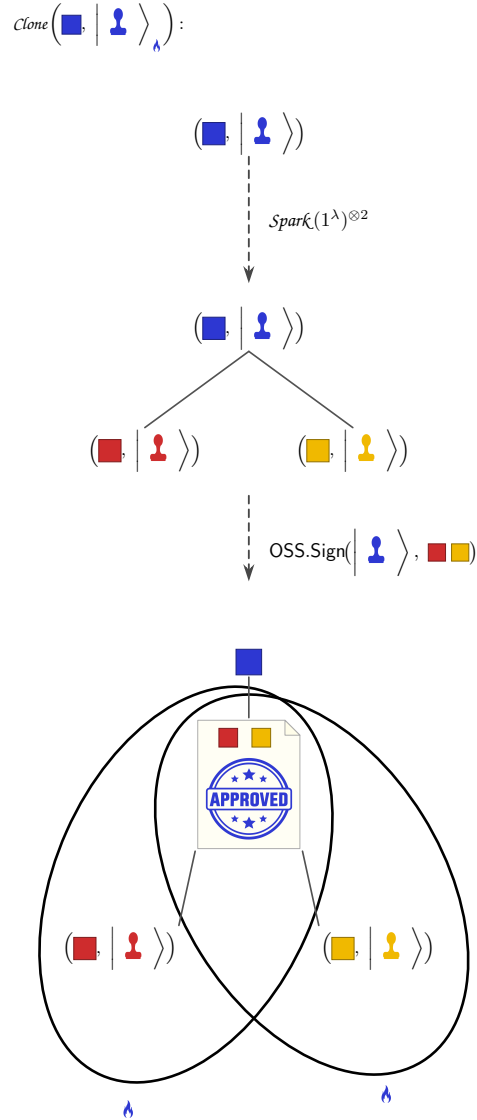}
    \caption{\textbf{Cloning a flame.} The black ellipses enclose each of the new blue flames. Note that although the new one-shot signatures are red and yellow, the two flames remain blue, corresponding to the original serial number.}
    \label{fig:clone-visual}
\end{figure}

\begin{figure}[htbp]
    \overfullrule=0pt
    \centering
    \input{figs/images/construction/verify/verify-main}
    \caption{\textbf{Verifying a flame.} Note that each honestly cloned flame is a branch in the tree, and thus the verification algorithm verifies the validity of the entire branch. We note that our graphical representation omits an important part of the verification process: checking that the length of the branch (its depth in the tree) does not exceed the allowed bound (logarithmic for \cref{const:fireFromOneShot}(a) and linear for \cref{const:fireFromOneShot}(b)).}
    \label{fig:verify-visual}
\end{figure}

\begin{figure}[htbp]
    \overfullrule=0pt
    \centering
    \input{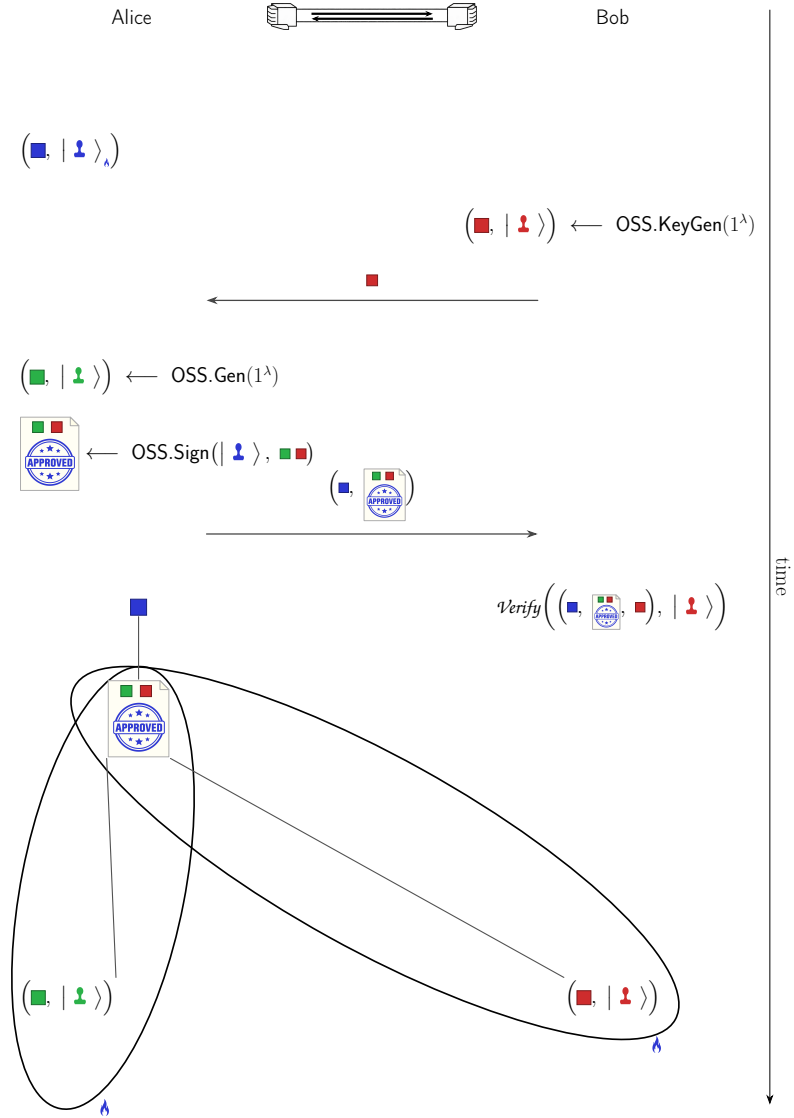}
    \caption{\textbf{Conversing a flame.} Note that at the end of the protocol the binary signatures tree has a similar form to the cloning tree, even though the leaves are held by different parties.}
    \label{fig:converse-visual}
\end{figure}

\end{document}